\documentclass[11pt,a4paper, leqno]{article}                    
\usepackage{graphicx}
\usepackage[centertags]{amsmath}
\usepackage{amsfonts}
\usepackage{amssymb}
\usepackage{amsthm}
\usepackage{epsfig}
\usepackage{setspace}
\usepackage{comment}
\usepackage{ae} 
\usepackage[authoryear]{natbib}
\usepackage{graphicx}

\theoremstyle{plain}
\newtheorem{thm}{Theorem}[section]
\newtheorem{lem}[thm]{Lemma}

\newtheorem{prop}[thm]{Proposition}
\theoremstyle{definition}

\theoremstyle{remark}

\newtheorem{rem}[thm]{Remark}

\newcommand{\E}{\operatorname{E}}
\renewcommand{\P}{\operatorname{P}}

\newcommand{\Var}{\operatorname{Var}}

\newcommand{\keywords}[1]{\par\addvspace\baselineskip\text{\it{Keywords}: } \noindent\ignorespaces#1}

\begin{document}
\title{Pricing Bermudan options using nonparametric regression: optimal rates of convergence for
lower
estimates}
\author{Denis Belomestny$^{1,\,}$\thanks{supported in part by the SFB 649 `Economic Risk'.
}}
\footnotetext[1]{Weierstrass Institute for Applied Analysis and
Stochastics, Mohrenstr. 39, 10117 Berlin, Germany.
{\tt{belomest@wias-berlin.de}}. }
\maketitle
\begin{abstract}
The problem of pricing Bermudan options using Monte Carlo and a nonparametric regression
is considered. We derive optimal non-asymptotic bounds for a
lower biased estimate based on the suboptimal stopping rule constructed using
some estimates of continuation values.
These estimates may be of different nature, they may be local or global,
with the only requirement being that the deviations of these estimates
from the true continuation values can be uniformly bounded in probability.
As an illustration, we discuss a class of local polynomial estimates which,
under some regularity conditions, yield continuation values estimates possessing this property.
\end{abstract}

\section{Introduction}
An American option grants the holder the right to select the time at which to exercise the option, and in this differs from a European option which may be exercised only at a fixed date.
A general class of American option  pricing problems can be formulated through
an \( \mathbb{R}^{d} \) Markov process \( \{X(t),\, 0\leq t \leq T\}   \)
defined on a filtered probability space \( (\Omega,\mathcal{F},(\mathcal{F}_{t})_{0\leq t\leq T},\P) \).
It is assumed that \( X(t) \) is adapted to \( (\mathcal{F}_{t})_{0\leq t\leq T} \) in the sense that each
\( X_{t} \) is \( \mathcal{F}_{t} \) measurable. Recall that each \( \mathcal{F}_{t} \)
is a \( \sigma \)-algebra of subsets of \( \Omega \) such  that
\( \mathcal{F}_{s}\subseteq \mathcal{F}_{t}\subseteq \mathcal{F} \) for \( s\leq t \). We
interpret \( \mathcal{F}_{t} \) as all relevant financial information available up to time \( t \). We restrict attention to options
admitting a finite set of exercise opportunities \( 0=t_{0}<t_{1}<t_{2}<\ldots<t_{L}=T \), sometimes called Bermudan options. If exercised at time \( t_{l},\, l=1,\ldots, L \), the option
pays \( f_{l}(X(t_l)) \), for some known functions \( f_{0}, f_{1},\ldots, f_{L} \) mapping
\( \mathbb{R}^{d} \) into \( [0,\infty) \). Let \( \mathcal{T}_{n} \) denote the set of stopping times taking values in \( \{ n,n+1,\ldots,L \} \).
A standard result in the theory of contingent claims states that the equilibrium price \( V_{n}(x) \)  of the American option  at time \( t_{n} \) in state \( x \) given that the option was not exercised prior to \( t_{n} \) is its value under an optimal exercise policy:
\begin{eqnarray*}
    V_{n}(x)=\sup_{\tau\in \mathcal{T}_{n}}\E[f_{\tau}(X(t_\tau))|X(t_n)=x],\quad x\in
    \mathbb{R}^{d}.
\end{eqnarray*}
Pricing an American option thus reduces to solving an optimal stopping problem.
Solving this optimal stopping problem and pricing an American option are
straightforward in low dimensions. However, many problems arising
in practice (see e.g. \citet{Gl}) have high dimensions, and these applications
have motivated the development of Monte Carlo methods for pricing American option.
Pricing American style derivatives with Monte Carlo is a challenging task because
the determination of optimal exercise strategies  requires a backwards dynamic
programming algorithm that appears to be incompatible with the forward nature
of Monte Carlo simulation. Much research was focused on the development of
fast methods to compute approximations to the optimal exercise policy.
Notable examples include the functional optimization approach in \citet{A},
mesh method of \citet{BG}, the regression-based approaches of \citet{Car}, \citet{LS}, \citet{TV} and \citet{E}.
A common feature of all above mentioned algorithms is that they deliver
estimates \( \widehat C_{0}(x),\ldots, \widehat C_{L-1}(x) \) for
the so called continuation values:
\begin{eqnarray}
\label{CV}
    C_{k}(x):=\E[V_{k+1}(X(t_{k+1}))|X(t_{k})=x], \quad k=0,\ldots,L-1.
\end{eqnarray}
An estimate for \( V_{0} \), the price of the option at time \( t_{0} \)
can then be defined as
\begin{eqnarray*}
    \widetilde V_{0}(x):=\max\{ f_{0}(x),\widehat C_{0}(x) \}, \quad x\in \mathbb{R}^{d}.
\end{eqnarray*}
This estimate  basically inherits all properties of \( \widehat C_{0}(x) \). In particular, it is
usually impossible  to determine the sign of the bias of \( \widetilde V_{0}\)
since the bias of \( \widehat C_{0}  \) may change its sign.
One way to get
a lower bound (low biased estimate) for \( V_{0} \)  is to  construct a (generally suboptimal) stopping rule
\begin{eqnarray*}
    \widehat\tau=\min\{0\leq k \leq L: \widehat C_{k}(X(t_{k}))\leq f_{k}(X(t_{k}))\}
\end{eqnarray*}
with \( \widehat C_{L}\equiv 0 \) by definition. Simulating a new
independent set of trajectories and averaging
the pay-offs stopped according to \( \widehat\tau \) on these trajectories gives us a lower bound
\( \widehat V_{0} \) for  \( V_{0} \).
As was observed by practitioners, the so constructed estimate \( \widehat V_{0} \) has rather stable behavior with
respect to  the estimates of continuation values \( \widehat C_{0}(x),\ldots, \widehat C_{L-1}(x) \), that is even rather
poor estimates of continuation values may lead to a good estimate \( \widehat V_{0}. \)
The aim of this paper is to find a theoretical explanation of this observation and
to investigate the properties of \( \widehat V_{0} \). In particular, we derive optimal
non-asymptotic bounds
for the bias \( V_{0}-\E\widehat V_{0} \)  assuming some uniform probabilistic bounds for \( C_{r}-\widehat C_{r}  \).
It is shown that the bounds for \( V_{0}-\E\widehat V_{0} \) are usually much tighter than ones
for \( V_{0}-\E\widetilde V_{0}  \) implying a better quality of \( \widehat V_{0} \) as
compared to the quality of \( \widetilde V_{0} \) constructed using  one and the same set
of estimates for continuation values. As an example, we consider the class of local polynomial estimators for
continuation values and derive explicit  convergence rates for \( \widehat V_{0} \) in this case.
\par
The issues of convergence for regression algorithms have been already studied in several papers.
\citet{CLP} were first who proved the convergence of
the Longstaff-Schwartz algorithm. \citet{GY} have shown that the
number of Monte Carlo paths has  to be in general exponential in the number of
basis functions used for regression in order to ensure convergence.
Recently, \citet{EKT}   have derived the rates of convergence
for continuation values estimates obtained by the so called dynamic
look-ahead algorithm (see Egloff (2004))  that
``interpolates'' between Longstaff-Schwartz and  Tsitsiklis-Roy algorithms.
As was shown in these papers   the  convergence rates for  \( \widetilde V_{0} \)  coincide with the
rates of \( \widehat C_{0} \) and are  determined by the smoothness
properties of the true continuation values \( C_{0},\ldots,C_{L-1} \). It turns out that the  convergence rates
for \( \widehat V_{0} \) depend  not only on the smoothness of continuation values (as opposite to \( \widetilde V_{0} \)), but
also on the behavior of the underlying process near the exercise boundary. Interestingly enough,
there are some cases where these rates  become almost  independent either of
the smoothness properties of \( \{ C_{k} \} \) or of the dimension of \( X \)
and the bias
of \( \widehat V_{0} \) decreases exponentially in the number of Monte Carlo
paths used to construct \( \{ \widehat C_{k} \} \).
\par
The paper is organized as follows. In Section~\ref{SBA} we introduce and discuss
the so called  boundary assumption
which describes the behavior of the underlying process \( X \) near
the exercise boundary
and heavily influences   the properties of
\( \widehat V_{0} \). In Section~\ref{SRC}
we derive non-asymptotic bounds for the bias \( V_{0}-\E\widehat V_{0} \) and prove that
these bounds are optimal in the minimax sense. In Section~\ref{SLP} we consider the class
of local polynomial estimates  and propose a sequential algorithm based on the
dynamic programming principle to estimate all continuation values. Finally, under
some regularity assumptions, we derive exponential bounds for the corresponding
continuation values estimates and consequently the bounds  for
the bias \( V_{0}-\E\widehat V_{0} \).

\section{Main results}
\subsection{Boundary assumption}
\label{SBA}
For the considered Bermudan option let us introduce
a continuation region $\mathcal{C}$ and  an exercise (stopping) region $%
\mathcal{E}$ :
\begin{eqnarray}
\mathcal{C} &:=&\left\{ (i,x):f_{i}(x)<C_{i}(x)\right\} ,  \label{U0}
\\
\mathcal{E} &:=&\left\{ (i,x):f_{i}(x)\geq C_{i}(x)\right\} .  \notag
\end{eqnarray}
Furthermore, let us assume that there exist constants $B_{0,k}>0$,
\( k=0,\ldots,L-1 \) and $\alpha >0$ such that the inequality
\begin{equation}
\label{BA}
\P_{t_{k}|t_{0}}(0<|C_{k}(X(t_k))-f_{k}(X(t_k))|\leq \delta)\leq B_{0,k}\delta^{\alpha },\quad \delta>0,
\end{equation}
holds for all \( k=0,\ldots,L-1 \), where \( \P_{t_{k}|t_{0}} \) is the conditional distribution
of \( X(t_{k}) \) given \( X(t_{0}) \).
Assumption \eqref{BA} provides a useful characterization of the behavior of the
continuation values $\{ C_{k} \}$ and payoffs \( \{ f_{k} \} \) near the exercise boundary \( \partial \mathcal{E} \).
Although this assumption seems quite natural to look at,
we make in this paper, to the best of our knowledge,  a first attempt to investigate
its influence  on the convergence rates of lower bounds
based on suboptimal stopping rules. We note that a similar condition, although much simpler, appears in the context of statistical classification problem (see, e.g. \citet{MT} and \citet{AT}).
\par
In the  situation when all functions \( C_{k}-f_{k},\, k=0,\ldots,L-1\) are smooth
and have
non-vanishing derivatives in the vicinity of the exercise boundary,
we have \( \alpha=1 \). Other values of \( \alpha \) are possible as well.
We illustrate this  by two simple examples.
\paragraph{Example 1}
Fix some \( \alpha>0 \) and consider a two period (\( L=1\))
Bermudan power put option with the payoffs
\begin{eqnarray}
\label{PCPayOff}
& f_{0}(x)=f_{1}(x)=(K^{1/\alpha}-x^{1/\alpha})^{+}, \quad x\in \mathbb{R}_{+},  \quad K>0.
\end{eqnarray}
Denote by \( \Delta  \) the length of the exercise period, i.e.
\( \Delta=t_{1}-t_{0}. \)
If the process  \( X \) follows the Black-Scholes model with volatility \( \sigma \) and zero
interest rate,
then one can show that
\begin{multline*}
C_{0}(x):=\E[f_{1}(X(t_{1}))|X(t_{0})=x]=
K^{1/\alpha}\Phi(-d_{2})
\\
-x^{1/\alpha}e^{\Delta (\alpha^{-1}-1)(\sigma^{2}/2\alpha)}
\Phi(-d_{1})
\end{multline*}
with \( \Phi \) being the cumulative distribution function of the standard normal distribution,
\begin{eqnarray*}
    d_{1}=\frac{\log(x/K)+\left( \frac{1}{\alpha}-\frac{1}{2} \right)\sigma^{2} \Delta  }{\sigma\sqrt{\Delta }}
\end{eqnarray*}
and \( d_{2}=d_{1}-\sigma\sqrt{\Delta }/\alpha. \)
As can be easily seen, the function \( C_{0}(x)-f_{0}(x) \) satisfies
\( |C_{0}(x)-f_{0}(x)|\asymp x^{1/\alpha} \) for \( x\to +0 \) and
\( C_{0}(x)>f_{0}(x) \) for all \( x>0 \) if \( \alpha\geq 1 \).
Hence
\begin{eqnarray*}
\P(0<|C_{0}(X(t_{0}))-f_{0}(X(t_{0}))|\leq \delta)\lesssim \delta^{\alpha }, \quad \delta\to
0, \quad \alpha\geq 1.
\end{eqnarray*}
Taking different \( \alpha \) in the definition of  the payoffs \eqref{PCPayOff}, we get \eqref{BA}
satisfied for \( \alpha \) ranging from \( 1 \) to \( \infty \).

\begin{figure}[pth]
\centering \includegraphics[width=8cm]{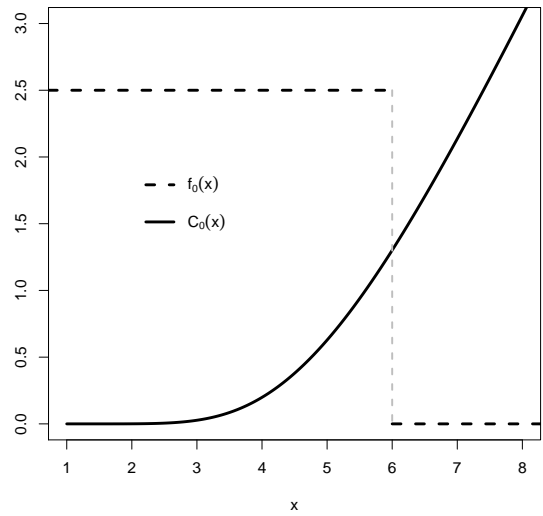}
\caption{ Illustration to  Example 2.}%
\label{BoundaryA}%
\end{figure}
In fact, even the extreme case ``\( \alpha=\infty \)''
may take place as shown in the next example.
\paragraph{Example 2}
Let us consider again a two period Bermudan option such that
the corresponding continuation value \( C_{0}(x)=\E[f_{1}(X(t_{1}))|X(t_{0})=x] \)
is positive and  monotone increasing function of \( x \) on any compact
set in \( \mathbb{R} \).
Fix some \( x_{0}\in \mathbb{R} \) and choose  \( \delta_{0} \) satisfying
\( \delta_{0}<C_{0}(x_{0}) \). Define the payoff function \( f_{0}(x) \)
in the following way
\begin{eqnarray*}
f_{0}(x)=
\begin{cases}
  C_{0}(x_{0})+\delta_{0},&\,x<x_{0},\\
  C_{0}(x_{0})-\delta_{0},&\,x\geq x_{0}.
\end{cases}
\end{eqnarray*}
So, \( f_{0}(x) \) has a ``digital'' structure.
Figure~\ref{BoundaryA} shows the plots of \( C_{0} \) and \( f_{0} \) in the case
where \( X \) follows the Black-Scholes model and \( f_{1}(x)=(x-K)^{+} \).
It is easy to see that
\begin{eqnarray*}
\P_{t_{0}}(0<|C_{0}(X(t_0))-f_{0}(X(t_0))|\leq \delta_{0})=0.
\end{eqnarray*}
On the other hand
\begin{eqnarray*}
\mathcal{C}&=&\{x\in \mathbb{R}: C_{0}(x)\geq f_{0}(x) \}=\{ x\in \mathbb{R}: x\geq
x_{0}\},
\\
\mathcal{E}&=&\{x\in \mathbb{R}: C_{0}(x)<f_{0}(x) \}=\{ x\in \mathbb{R}: x< x_{0}\}.
\end{eqnarray*}
So, both continuation and exercise regions are not trivial in this case.
\par
The last example is of particular interest because as will be shown
in the next sections the bias of \( \widehat V_{0} \)
decreases  in this case exponentially in the number of Monte Carlo paths used to
estimate the continuation values, the lower bound \( \widehat V_{0} \)  was constructed from.

\subsection{Non-asymptotic bounds for \( V_{0}-\E\widehat V_{0} \)}
\label{SRC}
Let \( \widehat C_{k,M}, \, k=1,\ldots,L-1, \) be some estimates of continuation
values obtained using \( M \)
paths of the underlying process \( X \) starting from \( x_{0} \) at time \( t_{0} \).
We may think of \( (X^{(1)}(t),\ldots,X^{(M)}(t))\) as being
a vector process on the product probability space with \( \sigma \)-algebra \( \mathcal{F}^{\otimes M}
\) and the product measure \( \P_{x_{0}}^{\otimes M} \) defined on \( \mathcal{F}^{\otimes M} \) via
\[
\P_{x_{0}}^{\otimes M}(A_{1}\times\ldots\times A_{M})=
\P_{x_{0}}(A_{1})\cdot\ldots \cdot \P_{x_{0}}(A_{M}),
\]
with  \( A_{m}\in \mathcal{F},\, m=1,\ldots,M \).
Thus, each \( \widehat C_{k,M}, \, k=0,\ldots, L-1, \) is measurable with respect to \( \mathcal{F}^{\otimes
M}\).
The following proposition provides non-asymptotic bounds for the bias
\( V_{0}-\E_{\P_{x_{0}}^{\otimes M}}[V_{0,M}] \)
given  uniform probabilistic bounds for \( \{ \widehat C_{k,M}\} \).
\begin{prop}
\label{CR}
Suppose that there exist constants \( B_{1}, \, B_{2} \) and a positive sequence \( \gamma_{M} \) such that
for any \( \delta>\delta_{0}>0 \) it holds
\begin{eqnarray}
\label{ExpB}
    \P_{x_{0}}^{\otimes M}
    \left(|\widehat C_{k,M}(x)-C_{k}(x)|\geq
    \delta\gamma^{-1/2}_{M}\right)\leq B_{1}\exp(-B_{2}\delta)
\end{eqnarray}
for almost all \( x \) with respect to
\( \P_{t_{k}|t_{0}} \), the conditional distribution
of \( X(t_{k}) \) given \( X(t_{0}) \), \(k=0,\ldots, L-1 \). Define
\begin{eqnarray}
    \label{LowerBound}
     V_{0,M}:=\E\left[f_{\widehat\tau_{M}}(X(t_{\widehat\tau_{M}}))|X(t_{0})=x_{0}\right]
\end{eqnarray}
with
\begin{eqnarray}
\label{StopRule}
    \widehat\tau_{M}:=\min\left\{0\leq k \leq L:
    \widehat C_{k,M}(X(t_{k}))\leq f_{k}(X(t_{k}))\right\}.
\end{eqnarray}
If the boundary condition \eqref{BA} is fulfilled, then
\begin{eqnarray*}
    0\leq V_{0}-\E_{\P_{x_{0}}^{\otimes M}}[V_{0,M}]\leq B \left[ \sum_{l=0}^{L-1}B_{0,l} \right]\gamma^{-(1+\alpha)/2}_{M}
\end{eqnarray*}
with some constant \( B \) depending only on \( \alpha \), \( B_{1} \) and \( B_{2} \).
\end{prop}
The above convergence rates can not be in general improved as shown in the next theorem.
\begin{prop}
\label{LowerBound}
Let \( L=2 \). Fix a pair of non-zero payoff functions \( f_{1},f_{2} \) such that
\( f_{2}: \mathbb{R}^{d}\to \{ 0,1 \} \) and \( 0<f_{1}(x)<1 \) on \( [0,1]^{d}. \)
Let \( \mathcal{P}_{\alpha} \) be a class of pricing  measures
such that the
boundary condition \eqref{BA} is fulfilled with some \( \alpha>0 \).
For any positive sequence \( \gamma_{M} \) satisfying
\[
  \gamma^{-1}_{M}=o(1), \quad \gamma_{M}=O(M), \quad M\to \infty,
\]
there exist a subset \( \mathcal{P}_{\alpha,\gamma}\) of \( \mathcal{P}_{\alpha} \)
and a constant \( B>0 \) such that for any \( M\geq 1 \), any
stopping rule \( \widehat\tau_{M} \) and any set of estimates
\( \{ \widehat C_{k,M} \} \)  measurable w.r.t. \( \mathcal{F}^{\otimes M} \),
we have for some \( \delta>0 \) and \( k=1,2, \)
\begin{eqnarray*}
 \sup_{\P\in \mathcal{P}_{\alpha,\gamma}}\P^{\otimes M}\left(|\widehat C_{k,M}(x)-C_{k}(x)|\geq \delta\gamma^{-1/2}_{M}\right)>0
\end{eqnarray*}
for almost all \( x  \) w.r.t. any \( \P\in \mathcal{P}_{\alpha,\gamma} \)
and
\begin{eqnarray*}
 \sup_{\P\in \mathcal{P}_{\alpha,\gamma}}\left\{ \sup_{\tau\in\mathcal{T}_{0}}\E^{\mathcal{F}_{t_{0}}}_{\P}[f_{\tau}(X(t_\tau))]-\E_{\P^{\otimes M}}[\E^{\mathcal{F}_{t_{0}}}_{\P}f_{\widehat\tau_{M}}(X(t_{\widehat\tau_{M}}))] \right\}\geq B
 \gamma^{-(1+\alpha)/2}_{M}.
\end{eqnarray*}
\end{prop}

Finally, we discuss the case when ``\( \alpha=\infty \)'',  meaning that there exists
\( \delta_{0}>0 \) such that
\begin{eqnarray}
\label{BAF}
\P_{t_{k}|t_{0}}(0<|C_{k}(X(t_k))-f_{k}(X(t_k))|\leq \delta_{0})=0
\end{eqnarray}
for \( k=0,\ldots,L-1. \) This is very favorable situation for the pricing of the corresponding Bermudan option. It turns out
that if  the continuation
values estimates \( \{ \widehat C_{k,M} \} \) satisfy a kind of exponential inequality
and \eqref{BAF} holds, then the bias of \(V_{0,M} \) converges to zero
exponentially fast in \( \gamma_{M} \).
\begin{prop}
\label{CRF}
Suppose that for any \( \delta>0 \) there exist constants \( B_{1}, \, B_{2} \)
possibly depending on \( \delta \) and a sequence of positive numbers \( \gamma_{M} \) not depending on \( \delta \)
such that
\begin{eqnarray}
\label{ExpBF}
    \P_{x_{0}}^{\otimes M}
    \left(|\widehat C_{k,M}(x)-C_{k}(x)|\geq
    \delta\right)\leq B_{1}\exp(-B_{2}\gamma_{M})
\end{eqnarray}
for almost all \( x \)
with respect to \( \P_{t_{k}|t_{0}} \), \(k=0,\ldots, L-1 \). Assume also that
there exists a constant \( B_{f}>0 \) such that
\begin{equation}
\E\left[\max_{k=0,\ldots,L}f^{2}_{k}(X(t_{k})) \right]\leq B_{f}.
\end{equation}
If the condition \eqref{BAF} is fulfilled with some \( \delta_{0}>0 \), then
\begin{eqnarray*}
    0\leq V_{0}-\E_{\P_{x_{0}}^{\otimes M}}[V_{0,M}]\leq B_{3} L\exp(-B_{4}\gamma_{M})
\end{eqnarray*}
with some constant \( B_{3} \) and \( B_{4} \) depending only on \( B_{1} \), \( B_{2} \)
and \( B_{f} \).
\end{prop}
\paragraph{Discussion}
Let us make a few remarks on the results of this section.
First, Proposition~\ref{CR} implies that the  convergence rates  of  \( \widehat V_{0,M} \), a Monte Carlo estimate for \( V_{0,M} \), are always faster than the convergence
rates of \( \{ \widehat C_{k,M} \} \) provided that \( \alpha>0 \). Indeed, while the
convergence rates of \( \{ \widehat C_{k,M} \} \) are of order \( \gamma_{M}^{-1/2} \),
the bias of \( \widehat V_{0,M} \) converges to zero as fast as \( \gamma_{M}^{-(1+\alpha)/2}. \)
As to the variance of \( \widehat V_{0,M} \), it can be made arbitrary small
by  averaging \( \widehat V_{0,M} \) over a large number of sets,
each consisting of \( M \) trajectories, and by taking a large number  of new independent Monte Carlo paths used to average the payoffs
stopped according to \( \widehat\tau_{M}. \)
\par
Second, if the condition \eqref{BAF} holds true, then the bias of \( \widehat V_{0,M} \)
decreases exponentially in \( \gamma_{M} \), indicating that
even very unprecise estimates of  continuation values would lead to
the estimate \( \widehat V_{0,M} \) of acceptable quality.
\par
Finally, let us stress that the results obtained in this section are quite general
and do not depend on the particular form of the estimates
\( \{ \widehat C_{k,M} \} \),
only the inequality \eqref{ExpB} being crucial for the results to hold. This
inequality
holds for various types of estimators. These may be global least squares estimators, neural networks (see \citet{KKT})
or local polynomial estimators.
The latter type of estimators has  not yet been  well investigated (see, however,
\citet{BMS} for some empirical results) in the context of pricing Bermudan option and we are going to fill this gap.
In the next sections we will show that if  all continuation values \( \{ C_{k} \} \) belong
to the H\"older class \( \Sigma(\beta,H,\mathbb{R}^{d}) \) and the conditional law of
\( X \) satisfies some regularity assumptions, then
local polynomial estimates of continuation values satisfy inequality \eqref{ExpB} with
\( \gamma_{M}=M^{2\beta/(2(\beta+\nu)+d)}\log^{-1}(M) \) for some \( \nu\geq 0 \).
\begin{rem}
In the case of projection estimates for continuation values,
some nice bounds were recently derived
in \citet{VR}.  Let \( \{ X_{k}, \, k=0,\ldots, L \}  \) be an ergodic Markov chain with the invariant distribution \( \pi  \) and \( f_{0}(x)\equiv\ldots\equiv f_{L}(x)\equiv f(x), \) then \( C_{0}\equiv\ldots \equiv C_{L-1}(x)=C(x), \) provided that \( X_{0} \) is distributed according to \( \pi  \).
 Furthermore, suppose that an estimate \( \widehat C(x) \) for the continuation value \( C(x) \) is available and satisfies a projected Bellman equation
 \begin{equation}
 \label{BE}
 \widehat C(x)=e^{-\rho}\Pi \E_{\pi}[\max\{ f(X_{1}), \widehat C(X_{1})) \}|X_{0}=x], \quad \rho>0,
 \end{equation}
 where \( \Pi \) is the corresponding projection operator. Define
\[ 
\widehat V_{0}(x):=\E\left[f_{\widehat\tau}(X_{\widehat\tau})|X_0=x\right] 
\]
with
\begin{eqnarray*}
    \widehat\tau:=\min\left\{0\leq k \leq L:
    \widehat C(X_{k})\leq f(X_{k})\right\},
\end{eqnarray*}
then as shown in \citet{VR}
\begin{eqnarray}
\label{VRI}
\left[ \E_{\pi }|V_{0}(X_{0})-\widehat V_{0}(X_{0})|^{2} \right]^{1/2} \leq D\left[ \E_{\pi }|C(X_{0})-\Pi C(X_{0})|^{2} \right]^{1/2}
\end{eqnarray}
with some absolute constant \( D \) depending on \( \rho  \) only. The inequality \eqref{VRI} indicates that the quantity
\[
\left[ \E_{\pi }|V_{0}(X_{0})-\widehat V_{0}(X_{0})|^{2} \right]^{1/2}
\]
might be much smaller than
\( \sup_{x}|C(x)-\widehat C(x)| \) and hence qualitatively supports the same sentiment as in our paper.
\end{rem}
\subsection{Local polynomial estimation}
\label{SLP}
We first introduce some notations related to local polynomial estimation.
Fix some \( k \) such that \( 0\leq k < L \) and suppose that
we want to estimate a regression function
\begin{eqnarray*}
    \theta_{k}(x):=\E[g(X(t_{k+1}))|X(t_{k})=x],\quad x\in \mathbb{R}^{d}
\end{eqnarray*}
with \( g: \mathbb{R}^{d}\to \mathbb{R} \). Consider \( M \) trajectories of the
process \( X \)
\begin{eqnarray*}
    (X^{(m)}(t_{0}),\ldots, X^{(m)}(t_{L})), \quad m=1,\ldots,M,
\end{eqnarray*}
all starting from \( x_{0} \), i.e.  \( X^{(1)}(t_{0})=\ldots=X^{(M)}(t_{0})=x_{0} \).
For some \( h>0 \), \( x\in \mathbb{R}^{d} \), an integer \( l\geq 0 \) and a function \( K:\mathbb{R}^{d}\to \mathbb{R}_{+} \), denote by \( q_{x,M} \) a polynomial on \( \mathbb{R}^{d} \) of degree \( l \) (maximal order of the
multi-index is less than or equal to \( l \)) which minimizes
\begin{equation}
\label{OF}
    \sum_{m=1}^{M}\left[Y^{(m)}(t_{k+1})-q_{x,M}(X^{(m)}(t_{k})-x)\right]^{2}K\left( \frac{X^{(m)}(t_{k})-x}{h}
    \right),
\end{equation}
where \( Y^{(m)}(t)=g(X^{(m)}(t)) \).
The local polynomial estimator \( \widehat \theta_{k,M}(x) \) of order \( l \) for the
value \( \theta_{k}(x) \)
of the regression function \( \theta_{k} \) at point \( x \) is defined as \( \widehat \theta_{k,M}(x)=q_{x,M}(0) \)
if \( q_{x,M} \) is the unique minimizer of \eqref{OF} and \( \widehat \theta_{k,M}(x)=0 \) otherwise.
The value \( h \) is called the bandwidth and the function \( K \) is called the
kernel of the local polynomial estimator.
\par
Let \( \pi_{u} \) denote the coefficients of \( q_{x,M} \) indexed by
the multi-index \( u\in \mathbb{N}^{d} \), \(q_{x,M}(z)=\sum_{|u|\leq l}\pi_{u}z^{u} \). Introduce the vectors
\( \Pi=(\pi_{u})_{|u|\leq l} \) and \( S=(S_{u})_{|u|\leq l} \) with
\begin{eqnarray*}
    S_{u}=\frac{1}{Mh^{d}}\sum_{m=1}^{M}Y^{(m)}(t_{k+1})\left(\frac{X^{(m)}(t_{k})-x}{h}\right)^{u}K\left( \frac{X^{(m)}(t_{k})-x}{h}
    \right).
\end{eqnarray*}
Let \( Z(z)=(z^{u})_{|u|\leq l}\) be the vector of all monomials of order less
than or equal to \( l \) and the matrix
\( \Gamma=(\Gamma_{u_{1},u_{2}})_{|u_{1}|,|u_{2}|\leq l} \) be defined as
\begin{equation}
\label{Gamma}
  \Gamma_{u_{1},u_{2}}=\frac{1}{Mh^{d}}\sum_{m=1}^{M}\left(\frac{X^{(m)}(t_{k})-x}{h}\right)^{u_{1}+u_{2}}K\left( \frac{X^{(m)}(t_{k})-x}{h} \right).
\end{equation}
The following result is straightforward.
\begin{prop}
\label{LPR}
If the matrix \( \Gamma \) is positive definite, then there exists a unique polynomial on \( \mathbb{R}^{d} \) of degree \( l \) minimizing \eqref{OF}. Its vector of coefficients is given by
\( \Pi=\Gamma^{-1}S \) and the corresponding local polynomial regression function estimator has the form
\begin{multline}
    \widehat \theta_{k,M}(x)=Z^{\top}(0)\Gamma^{-1} S
    \\
    =\frac{1}{Mh^{d}}\sum_{m=1}^{M}Y^{(m)}(t_{k+1})K\left( \frac{X^{(m)}(t_{k})-x}{h} \right)
    \\
    \label{LPE}
    \times Z^{\top}(0)\Gamma^{-1} Z\left(\frac{X^{(m)}(t_{k})-x}{h}\right).
\end{multline}
\end{prop}
\begin{rem}
From the inspection of \eqref{LPE} it becomes  clear that any local polynomial
estimator can be represented as a weighted average of
 the ``observations'' \( Y^{(m)}, \, m=1,\ldots, M, \) with a special weights structure. Hence,
local polynomial estimators belong to the class of mesh estimators
introduced by \citet{BG} \cite[see also][Ch. 8]{Gl}.  Our results will show that
this particular type of mesh estimators has nice convergence properties in the class
of smooth continuation values.
\end{rem}
\subsection{Estimation algorithm for the continuation values}
\label{EstAlg}
According to the dynamic programming principle, the optimal continuation values \eqref{CV}
satisfy the following backward recursion
\begin{eqnarray*}
    C_{L}(x)&=&0,
    \\
    C_{k}(x)&=&\E[\max(f_{k+1}(X(t_{k+1})),C_{k+1}(X(t_{k+1})))|X(t_{k})=x], \quad x\in \mathbb{R}^{d}
\end{eqnarray*}
with \( k=1,\ldots,L-1 \).
Consider \( M \) paths of the process \( X \), all starting from \( x_{0}, \)
and  define estimates \( \widehat C_{1,M}, \ldots, \widehat C_{L,M} \) recursively in
the following way. First, we put \( \widehat C_{L,M}(x)\equiv 0 \).
Further, if an estimate of \( \widehat C_{k+1,M}(x) \) is already constructed we
define  \( \widehat C_{k,M}(x) \) as the local polynomial estimate of the function
\begin{equation}
\label{CTilde}
\widetilde C_{k,M}(x):=\E[\max(f_{k+1}(X(t_{k+1})),\widehat C_{k+1,M}(X(t_{k+1})))|X(t_{k})=x],
\end{equation}
based on  the sample
\[
 (X^{(m)}(t_{k}),\widehat C_{k+1,M}(X^{(m)}(t_{k+1}))),\quad m=1,\ldots,M.
\]
Note that all \( \widetilde C_{k,M} \) are \( \mathcal{F}^{\otimes M} \) measurable
random variables because the expectation in \eqref{CTilde} is taken with respect to
a new \( \sigma \)-algebra \( \mathcal{F} \) which is independent of \( \mathcal{F}^{\otimes M} \)
(one can start with the enlarged product \( \sigma \)-algebra \( \mathcal{F}^{\otimes (M+1)} \)
and take expectation in \eqref{CTilde} w.r.t. the first coordinate).
The main problem arising by the convergence analysis of the estimate \(  \widehat C_{k+1,M}  \)
is that all errors coming from the previous estimates \(  \widehat C_{j,M}, \, j\leq k \)
have to be taken into account. This problem has been already encountered by
\citet{CLP} who investigated the convergence of the Longstaff-Schwartz algorithm.

\subsection{Rates of convergence for \( V_{0}-\E\widehat V_{0} \)}
Let \( \beta>0 \). Denote by \( \lfloor \beta \rfloor \) the maximal integer that is strictly
less than \( \beta \). For any \( x\in \mathbb{R}^{d} \) and any \( \lfloor \beta \rfloor \)
times continuously differentiable real-valued function \( g \) on \( \mathbb{R}^{d} \), we denote
by \( g_{x} \) its Taylor polynomial of degree \( \lfloor \beta \rfloor \) at point \( x \)
\begin{eqnarray*}
    g_{x}(x')=\sum_{|s|\leq \lfloor \beta \rfloor}\frac{(x'-x)^{s}}{s!}D^{s}g(x),
\end{eqnarray*}
where \( s=(s_{1},\ldots,s_{d}) \) is a multi-index, \( |s|=s_{1}+\ldots+s_{d} \) and \( D^{s} \) denotes the differential
operator \( D^{s}=\frac{\partial^{s_{1}+\ldots+s_{d}}}{\partial x_{1}^{s_{1}}\cdot\ldots\cdot \partial x_{d}^{s_{d}}} \).
Let \( H>0 \). The class of \( (\beta,H,\mathbb{R}^{d}) \)-H\"older smooth functions, denoted
by \( \Sigma(\beta,H,\mathbb{R}^{d}) \), is defined as the set of functions \( g:\mathbb{R}^{d}\to \mathbb{R} \) that are \( \lfloor \beta \rfloor \) times continuously differentiable and satisfy,
for any \( x,x'\in \mathbb{R}^{d} \), the inequality
\begin{eqnarray*}
    |g(x')-g_{x}(x')|\leq H \| x-x' \|^{\beta}, \quad x'\in \mathbb{R}^{d}.
\end{eqnarray*}
Let us make two assumptions on the process \( X \)
\begin{description}
\item[(AX0)] There exists a bounded set \( \mathcal{A}\subset \mathbb{R}^{d} \) such that
\( \P(X(t_{0})\in \mathcal{A})=1 \) and \( \P_{s|t}(X(s)\in \mathcal{A})=1 \) for all \( t \) and \( s \) satisfying \( t_{0}\leq t\leq s\leq T. \)

\item[(AX1)] All transitional densities \( p(t_{k+1},y|t_{k},x), \, k=0,\ldots,L-1, \) of the process \( X \)
are uniformly bounded on \( \mathcal{A}\times \mathcal{A} \) and belong to the H\"older class \( \Sigma(\beta,H,\mathbb{R}^{d}) \) as functions
of \( x\in \mathcal{A} \), i.e. there exists \( \beta>1 \) with \( \beta-\lfloor \beta  \rfloor >0 \)
and a constant \( H \) such that the inequality
\begin{eqnarray*}
   |p(t_{k+1},y|t_{k},x')-p_{x}(t_{k+1},y|t_{k},x')|\leq H \| x-x' \|^{\beta}
\end{eqnarray*}
holds for all \( x,x',y\in \mathcal{A} \) and \( k=0,\ldots,L-1. \)
\end{description}
Consider a matrix valued
function \( \bar\Gamma(s,x)=(\Gamma_{u_{1},u_{2}})_{|u_{1}|,|u_{2}|
\leq \lfloor\beta\rfloor} \)
with elements
\[
 \bar\Gamma_{u_{1},u_{2}}(s,x):=\int_{\mathbb{R}^{d}}z^{u_{1}+u_{2}}K(z)p(s,x+hz|t_{0},x_{0})\, dz,
\]
for any \( s>t_{0}. \)
\begin{description}
  \item[(AX2)] We assume that the minimal eigenvalue of \(  \bar\Gamma \) satisfies
  \begin{eqnarray*}
   \label{EV_WGW}
   \min_{k=1,\ldots,L}\inf_{x\in \mathcal{A}}
   \min_{\| W \|=1}\left[ W^{\top}\bar\Gamma(t_{k},x) W \right]\geq\gamma_{0} h^{\nu}
  \end{eqnarray*}
  with some \( \nu\geq 0 \) and \( \gamma_{0}>0.  \)
\end{description}
Moreover, we shall assume that the kernel \( K \) fulfils the
following conditions
\begin{description}
  \item[(AK1)] \( K \) integrates to \( 1 \) on \( \mathbb{R}^{d} \) and
\begin{eqnarray*}
    \int_{\mathbb{R}^{d}}(1+\| u \|^{4\beta})K(u)\,du<\infty, \quad \sup_{u\in \mathbb{R}^{d}}(1+\| u \|^{2\beta})K(u)<\infty.
\end{eqnarray*}
\item[(AK2)] \( K \) is in the linear span
(the set of finite linear combinations) of functions \( k\geq 0 \) satisfying
the following property: the subgraph of \( k, \) \( \{ (s,u):\, k(s)\geq u \}, \)
can be represented as a finite number of Boolean operations among the sets of the form
\( \{ (s,u):\, p(s,u)\geq f(u) \} \), where  \( p \) is a polynomial on \( \mathbb{R}^{d}\times \mathbb{R} \)
and \( f \) is an arbitrary real function.
\end{description}

\paragraph{Discussion}
The assumption (AX0) may seem rather restrictive. In fact, as mentioned in
\citet{EKT},
one can always use a kind of ``killing'' procedure to localize process \( X \) to
a ball \( \mathcal{B}_{R} \) in \( \mathbb{R}^{d} \) around \( x_{0} \) of radius \( R \) . Indeed, one can replace process \( X(t) \) with
the process \( X^{\mathcal{K}}(t) \) killed at first exit time from  \( \mathcal{B}_{R} \).  This new process
\( X^{\mathcal{K}}(t) \) is again a Markov process and is connected to the original process \( X(t) \)
via the identity
\begin{eqnarray*}
    \E[g(X^{\mathcal{K}}(s))|X^{\mathcal{K}}(t)=x]=\E[g(X(s))M(s)|X(t)=x], \quad s>t,
\end{eqnarray*}
that holds for any integrable  \( g: \mathbb{R}^{d}\to \mathbb{R} \) with \( M(s)=\mathbf{1}(\tau_{R}>s) \) and \( \tau_{R}=\inf\{ t>0:\,X(t)\not \in \mathcal{B}_{R} \} \).
This implies that
\begin{multline}
\label{KI}
   \sup_{\tau\in \mathcal{T}_{0}}\left|\E^{\mathcal{F}_{t_{0}}}[f_{\tau}(X(t_\tau))]-\E^{\mathcal{F}_{t_{0}}}[f_{\tau}(X^{\mathcal{K}}(t_\tau))]\right|
   \\
   \leq \sup_{\tau\in \mathcal{T}_{0}}\left| \E^{\mathcal{F}_{t_{0}}}[f_{\tau}(X(t_\tau))\mathbf{1}(m_{\tau}>R)] \right|
\end{multline}
with \( m_{t}=\sup_{0\leq s \leq t}\| X(s)-x_{0} \| \). The r.h.s of \eqref{KI}
can be made arbitrary small by taking large values of \( R \) (the exact convergence
rates depend, of course, on the properties of the process \( X \)).
\par
Instead of ``killing'' the process \( X(t) \) upon leaving \( \mathcal{B}_{R} \) one can reflect it on the boundary of \( \mathcal{B}_{R}. \)
As can be seen a new reflected process \( X^{\mathcal{R}}(t) \) satisfies \eqref{KI} as well.
\paragraph{Example}
Let process \( X(t) \) be a \( d \)-dimensional diffusion process satisfying
\begin{eqnarray*}
    X(t)=x_{0}+\int_{t_{0}}^{t}\mu(X(t))\, dt+\int_{t_{0}}^{t}\sigma(X(t))\, dW(t), \quad t\geq t_{0}.
\end{eqnarray*}
Denote
by \( p^{\mathcal{K}}(s-t,y|x) \) the transition density of  the process \( X^{\mathcal{K}}\). Assume that a drift coefficient \( \mu  \) and a diffusion coefficient \( \sigma  \) are regular enough and \( \sigma  \) satisfies the so called uniform ellipticity condition on compacts, i. e. for each compact set \( K\subset  \mathbb{R}^{d}\)
\begin{description}
  \item[(AD1)]  \( \mu(\cdot)\in C_{b}^{k}(K) \) and \( \sigma(\cdot)\in C_{b}^{k}(K)  \) for some natural \( k>1, \)
  \item[(AD2)]  there is \( \sigma_{K}>0  \) such that for any \( \xi \in \mathbb{R}^{d} \)
  it holds
  \[
   \sum_{j,k=1}^{d}( \sigma(x)\sigma^{\top}(x))_{jk}\xi_{j}\xi_{k}  \geq \sigma_{K}\|\xi \|^{2}, \quad x\in K.
  \]
\end{description}
Then (see e.g. \citet{F}) for any fixed \( s>0 \),  \(p^{\mathcal{K}} (s,y|x) \) is a \( C^{k} ( \overline{\mathcal{B}}_{R}\times \overline{\mathcal{B}}_{R})  \) function in \( (x,y) \). Moreover, as shown in \citet{KS} (see also \citet{B}) under assumptions (AD1) and (AD2) there exist positive constants \( C_{i}, \, i=1,\ldots,4, \) such that
\begin{eqnarray*}
    C_{1}\psi_{\mathcal{K}}(s,x,y)s^{-d/2}e^{-C_{2}\| x-y \|^{2}/s}\leq p^{\mathcal{K}} (s,y|x)\leq C_{3}\psi_{\mathcal{K}}(s,x,y)s^{-d/2}e^{-C_{4}\| x-y \|^{2}/s}
\end{eqnarray*}
for all \( (s,x,y)\in (0,T]\times \mathcal{B}_{R}\times \mathcal{B}_{R}, \)
where
\[
\psi_{\mathcal{K}}(s,x,y):=\left(1\wedge \frac{(R-\| x -x_{0}\|)}{\sqrt{s}}\right)\left(1\wedge \frac{(R-\| y -x_{0}\|)}{\sqrt{s}}\right).
\]
Let us check now assumption (AX2) in the case when \( K(z)=\frac{\Gamma(1+d/2)}{\pi ^{d/2}}\mathbf{1}_{\{ \|z\|\leq1 \}} \). We have for any fixed \( s>t_{0} \) and \( W\in \mathbb{R}^{D} \) with \( D=d(d+1)\cdot\ldots\cdot(d+\lfloor \beta \rfloor-1)/\lfloor \beta \rfloor ! \)
\begin{eqnarray*}
  W^{\top}\bar\Gamma(s,x) W&=&\int_{\mathbb{R}^{d}}\left( \sum_{|\alpha|\leq \lfloor \beta \rfloor} W^{\alpha }z_{\alpha } \right)^{2} K(z)p^{\mathcal{K}}(s-t_{0},x+hz|x_{0})\, dz
  \\
  &\geq&
  B\int_{\mathcal{S}(x,R)}\left( \sum_{|\alpha |\leq \lfloor \beta \rfloor} W^{\alpha }z_{\alpha } \right)^{2}(R-\| x+hz-x_{0} \|)\, dz
\end{eqnarray*}
with some positive constant \( B \) depending on \( s-t_{0} \) and \( R \), and  \( \mathcal{S}(x,R):=\{ z: \| z \|\leq 1,\, \| x+hz-x_{0} \|\leq R \}. \)
Introduce
\begin{eqnarray*}
\widetilde {\mathcal{S}}(x,R):=\{ z: \| z \|\leq 1,\, \| x+hz -x_{0}\|\leq R-h/2 \}.
\end{eqnarray*}
Since \( \widetilde {\mathcal{S}}(x,R)\subset \mathcal{S}(x,R) \) we get
\begin{eqnarray*}
\int_{\mathcal{S}(x,R)}\left( \sum_{|\alpha|\leq \lfloor \beta \rfloor} W^{\alpha }z_{\alpha } \right)^{2}(R-\| x+hz-x_{0} \|)\, dz
\geq \frac{h}{2}\int_{\widetilde {\mathcal{S}}(x,R)}\left( \sum_{|\alpha |\leq \lfloor \beta \rfloor} W^{\alpha }z_{\alpha } \right)^{2}\, dz.
\end{eqnarray*}
Using now the fact that the Lebesgue  measure of the set \( \widetilde {\mathcal{S}}(x,R) \)  is larger than some positive number \( \lambda \) for all \( x\in \mathcal{B}_{R}, \) where \( \lambda  \) depends on \( R \) and \( d \) but does not depend on \( h, \) we get
\begin{eqnarray*}
       \min_{k=1,\ldots, L}\inf_{x\in \mathcal{B}_{R}}\left[ W^{\top}\bar\Gamma(t_{k},x) W \right]\geq \frac{Bh}{2}\inf_{\| W \|=1}\inf_{\mathcal{S}: |\mathcal{S}|>\lambda}\int_{\mathcal{S}}\left( \sum_{|\alpha |\leq \lfloor \beta \rfloor} W^{\alpha }z_{\alpha } \right)^{2}\, dz\geq \gamma_{0}h
\end{eqnarray*}
by the compactness argument. Thus, assumption (AX2) is fulfilled with \( \nu =1. \)
\par
Let us now reflect the diffusion process \( X(t) \) instead of ``killing'' it by defining
a reflected process \( X^{\mathcal{R}}(t) \) which satisfies  a reflected stochastic differential equation  in \( \mathcal{B}_{R} \),
with
oblique reflection at the boundary of \( \mathcal{B}_{R} \) in the conormal direction, i.e.
\begin{eqnarray*}
  X^{\mathcal{R}}(t)=x_{0}+\int_{t_{0}}^{t}\mu(X^{\mathcal{R}}(t))\, dt+\int_{t_{0}}^{t}\sigma(X^{\mathcal{R}}(t))\, dW(t)+\int_{t_{0}}^{t}\mathbf{n} (X^{\mathcal{R}}(t))\, dL(t), 
\end{eqnarray*}
where \( \mathbf{n} \) is the inward normal vector on the boundary of \( \mathcal{B}_{R} \) and \( L(t) \) is a local time process which increases only on  \( \{ \| x \|=R \}, \)  i.e.
\( L(t)=\int_{t_{0}}^{t}\mathbf{1}_{\{ \| X_{s} \|=R \}} \, dL(s).\)
Denote by  \( p^{\mathcal{R}} (s,y|x) \) a transition density of \(  X^{\mathcal{R}}(t) \).
It satisfies a parabolic  partial differential equation with Neumann boundary conditions.
Under (AD1) it belongs to \( C^{k} ( \overline{\mathcal{B}}_{R}\times \overline{\mathcal{B}}_{R})  \)
(see \citet{SU}) for any fixed \( s>0. \)
Moreover, using
a strong version
of the maximum principle \citep[see, e.g.][Theorem 1 in Chapter 2]{F}
one can show that under assumption (AD2) the transition density \( p^{\mathcal{R}} (s,y|x) \) is strictly positive on \( (0,T]\times \mathcal{B}_{R}\times \mathcal{B}_{R} \). Similar calculations as before show that in this case
\begin{eqnarray*}
    \min_{k=1,\ldots, L}\inf_{x\in \mathcal{B}_{R}}\left[ W^{\top}\bar\Gamma(t_{k},x) W \right]\geq\gamma_{0}>0
\end{eqnarray*}
and hence assumption (AX2) holds with \( \nu =0. \)
\begin{rem}
It can be shown that (AK2) is fulfilled if  \( K(x)=f(p(x)) \) for some
polynomial \( p \) and a bounded real function \( f \)  of bounded variation.
Obviously, the standard Gaussian kernel falls into this category. Another example
is the case where \( K \) is a pyramid or \( K=\mathbf{1}_{[-1,1]^{d}} \).
\end{rem}
In the sequel we will consider a truncated version of the
local polynomial estimator
\( \widehat C_{k,M}(x) \) which is defined as follows. If the smallest eigenvalue of the
matrix \( \Gamma \) defined in \eqref{Gamma} is greater than \( h^{\nu}(\log M)^{-1} \)
we set \( T[\widehat C_{k,M}](x) \) to be equal to the projection of \( \widehat C_{k,M}(x) \)
on the interval \( [0,C_{\max}] \) with \( C_{\max}=\max_{k=0,\ldots,L-1}\sup_{x\in \mathcal{A}}
C_{k}(x) \) (\( C_{\max} \) is finite due to (AX0) and (AX1)).  Otherwise, we put
\( T[\widehat C_{k,M}](x)=0 \).
The following propositions provide exponential bounds for  the truncated
estimator \( \{ T[\widehat C_{k,M}] \}\).
\begin{prop}
\label{ExpBoundsKern}
Let condition (AX0)-(AX2),(AK1) and (AK2) be satisfied  and let
\( \{ T[\widehat C_{k,M}] \} \)
be the continuation values estimates constructed as described in Section~\ref{EstAlg} using
truncated local polynomial estimators of degree \( \lfloor \beta \rfloor \). Then
there exist positive constants \( B_{1} \), \( B_{2} \) and \( B_{3} \) such that
for any \( h \) satisfying  \( B_{1}h^{\beta}<\sqrt{|\log h|/Mh^{d}} \) and any
\( \zeta\geq\zeta_{0} \) with some \( \zeta_{0}>0 \)
it holds
\begin{eqnarray*}
    &&\P_{x_{0}}^{\otimes M}\left(\sup_{x\in \mathcal{A}}
    |T[\widehat C_{k,M}](x)- C_{k}(x)|\geq   \zeta\sqrt{\frac{|\log h|}{Mh^{d+2\nu }}}\right)
    \leq B_{2}\exp(-B_{3}\zeta)
\end{eqnarray*}
for  \( k=1,\ldots,L-1 \).
As a consequence, we get with \( h=M^{-1/(2(\beta+\nu )+d)} \)
and any \( \zeta>\zeta_{0}>0 \)
\begin{eqnarray*}
   \P_{x_{0}}^{\otimes M}\left(\sup_{x\in \mathcal{A}}
   |T[\widehat C_{k,M}](x)- C_{k}(x)|\geq
   \frac{\zeta\log^{1/2} M}{M^{\beta/(2(\beta+\nu )+d)}}\right)\leq B_{2}\exp(-B_{3}\zeta).
\end{eqnarray*}
\end{prop}
\begin{prop}
Let condition (AX0)-(AX2),(AK1) and (AK2) be satisfied, then
for any \( \delta>0 \) there exist positive
constants \( B_{4} \) and \( B_{5} \) such that
\label{ExpBoundKernF}
\begin{eqnarray*}
   \P_{x_{0}}^{\otimes M}\left(\sup_{x\in \mathcal{A}}
    |T[\widehat C_{k,M}](x)- C_{k}(x)|\geq  \delta \right)
    \leq B_{4}\exp(-B_{5}M)
\end{eqnarray*}
for \( k=1,\ldots,L-1. \)
\end{prop}
\begin{rem}
As can be seen from the proof of  Proposition~\ref{ExpBoundsKern} and  Remark~\ref{DepVCD} (note that \( \omega  \) in \eqref{VC} grows linearly in \( d \) ) the constant \( B_{3} \)   decreases with
the dimension \( d \) as fast as  \( 1/d. \) The constant \( B_{5} \) is of order \( \delta_{0}^{(d+\nu)/\beta}/d.  \)
\end{rem}
Combining Proposition~\ref{CR} with Proposition~\ref{ExpBoundsKern} and Proposition~\ref{ExpBoundKernF} leads to the following
\begin{thm}
\label{MainR}
Let conditions (AX0)-(AX2), (AK1) and (AK2) be satisfied. Define
\begin{eqnarray*}
 V_{0,M}:=\E(f_{\widehat\tau_{M}}(X(t_{\widehat\tau_{M}}))|X(t_{0})=x_{0}),
\end{eqnarray*}
with
\begin{eqnarray*}
    \widehat\tau_{M}:=\min\{0\leq k \leq L: T[\widehat C_{k,M}](X(t_{k}))\leq
    f_{k}(X(t_{k}))\},
\end{eqnarray*}
where \( \{ T[\widehat C_{k,M}] \} \)
are   continuation values estimates constructed  using
truncated local polynomial estimators of degree \( \lfloor \beta \rfloor \).
If the boundary condition \eqref{BA} is fulfilled for some \( \alpha>0 \), then
\begin{eqnarray*}
    0\leq V_{0}-\E_{\P_{x_{0}}^{\otimes M}}[V_{0,M}]\leq D_{1}
    M^{-\beta(1+\alpha)/(2(\beta+\nu )+d)}\log^{(1+\alpha)/2}(M),
\end{eqnarray*}
with some constant \( D_{1} \). On the other hand, if the condition \eqref{BAF} is satisfied with
some \( \delta_{0}>0 \), then the bias of \( \widehat V_{0,M} \) decreases
exponentially in \( M \),
i.e. there exist positive constants \( D_{2}\) and \( D_{3} \),  such that
\begin{eqnarray*}
    0\leq V_{0}-\E_{\P_{x_{0}}^{\otimes M}}[V_{0,M}]\leq D_{2} \exp(-D_{3}M).
\end{eqnarray*}
\end{thm}
\paragraph{Discussion}
As we can see, the rates of convergence for \( \{ \widehat C_{k,M} \} \) are of
order
\[
 M^{-\beta/(2(\beta+\nu )+d)}\log^{1/2} M
\]
which can be proved  to be optimal under assumption (AX2), up to
a logarithmic factor, for the class of H\"older smooth
continuation values \( \{ C_{k}(x) \} \). On the other hand, the rates of convergence
for \( \E_{\P_{x_{0}}^{\otimes M}}[V_{0,M}] \) are of order
\[
 M^{-\beta(1+\alpha)/(2(\beta+\nu )+d)}\log^{(1+\alpha)/2} (M)
\]
and are always faster than ones of \( \{ \widehat C_{k,M} \} \) provided that
\( \alpha>0 \). The most interesting behavior of the lower bound \( \widehat V_{0,M} \)
can be observed if the condition
\eqref{BAF} is fulfilled. In this case the bias of
\( \widehat V_{0,M} \) becomes as small as \( \exp(-D_{3}M) \).  This means
that even
in the class of continuation values with an arbitrary low (but positive) H\"older smoothness
(e.g. in the class of non-differentiable continuation values)
and therefore with an arbitrary slow convergence rates of the estimates
\( \{ \widehat C_{k,M} \},\) the bias of the lower bound \( \widehat V_{0,M} \)
converges exponentially fast to zero.
\section{Numerical example: Bermudan max call}

This is a benchmark example studied in \citet{BG} and \citet{Gl}
among others. Specifically, the model with $d$ identically distributed
assets is considered, where each underlying has dividend yield $\delta $.
The risk-neutral dynamic of assets is given by
\begin{equation*}
\frac{dX_{k}(t)}{X_{k}(t)}=(r-\delta )dt+\sigma dW_{k}(t),\quad k=1,...,d,
\end{equation*}%
where $W_{k}(t),\,k=1,...,d$, are independent one-dimensional Brownian
motions and $r,\delta ,\sigma $ are constants. At any time $t\in
\{t_{0},...,t_{L}\}$ the holder of the option may exercise it and
receive the payoff
\begin{equation*}
f(X(t))=(\max (X_{1}(t),...,X_{d}(t))-\kappa)^{+}.
\end{equation*}%
We take \( d=2 \), \( r=5\% \), \( \delta=10\% \), \( \sigma=0.2 \), \( \kappa=100 \) and $t_{i}=iT/L,\,i=0,...,L$, with $T=3,\,L
=9$ as in \citet[Chapter 8]{Gl}. First, we estimate all continuation values using the dynamic programming
algorithm and the so called  Nadaraya-Watson regression estimator
\begin{equation}
\label{KE}\widehat C_{k,M}(x)=\frac{\sum_{m=1}^{M}K((x-X^{(m)}(t_{k})%
)/h)Y_{k+1}^{(m)}}%
{\sum_{m=1}^{M}K((x-X^{(m)}(t_{k}))/h)}
\end{equation}
with \( Y_{k+1}^{(m)}=\max(f_{k+1}(X^{(m)}(t_{k+1})),
e^{-rT/L}\widehat C_{k+1,M}(X^{(m)}(t_{k+1}))), \) \( k=0,\ldots,L-1. \)
Here \( K \) is a kernel, \( h>0 \) is a bandwidth and
\( (X^{(m)}(t_{1}),\ldots,X^{(m)}(t_{L})), \) \( m=1,\ldots,M, \) is a set of paths
of the process \( X \), all starting from the point \( x_{0}=(90,90) \) at \( t_{0}=0 \). As can be easily seen the
estimator \eqref{KE} is a
local polynomial estimator of degree \( 0 \).
Upon estimating \( \widehat C_{1,M} \), we define a first estimate for
the price of the option at time \( t_{0}=0 \)
as
\[
\widetilde V_{0}:=\frac{1}{M}\sum_{m=1}^{M}Y_{1}^{(m)}.
\]
Next, using the previously constructed estimates of continuation values,
we pathwise compute a stopping
policy \( \widehat \tau \)   via
\begin{eqnarray*}
\widehat\tau^{(n)}:=\min\left\{1\leq k \leq L: \widehat C_{k,M}(\widetilde X^{(n)}(t_{k}))\leq
    f_{k}(\widetilde X^{(n)}(t_{k}))\right\}, \quad n=1,\ldots,N,
\end{eqnarray*}
where  \( (\widetilde X^{(n)}(t_{1}),\ldots, \widetilde X^{(n)}(t_{L})), \)
\( n=1,\ldots,N, \) is a new independent set of trajectories of the process \( X \), all starting
from \( x_{0}=(90,90) \) at \( t_{0}=0 \).
The stopping policy \( \widehat \tau \) yields a lower bound
\begin{eqnarray*}
\widehat V_{0}=\frac{1}{N}\sum_{n=1}^{N}e^{-rt_{\widehat\tau^{(n)}}}f_{\widehat\tau^{(n)}}
(\widetilde X^{(n)}(t_{\widehat\tau^{(n)}})).
\end{eqnarray*}
\begin{figure}[pth]
\centering \includegraphics[width=12cm]{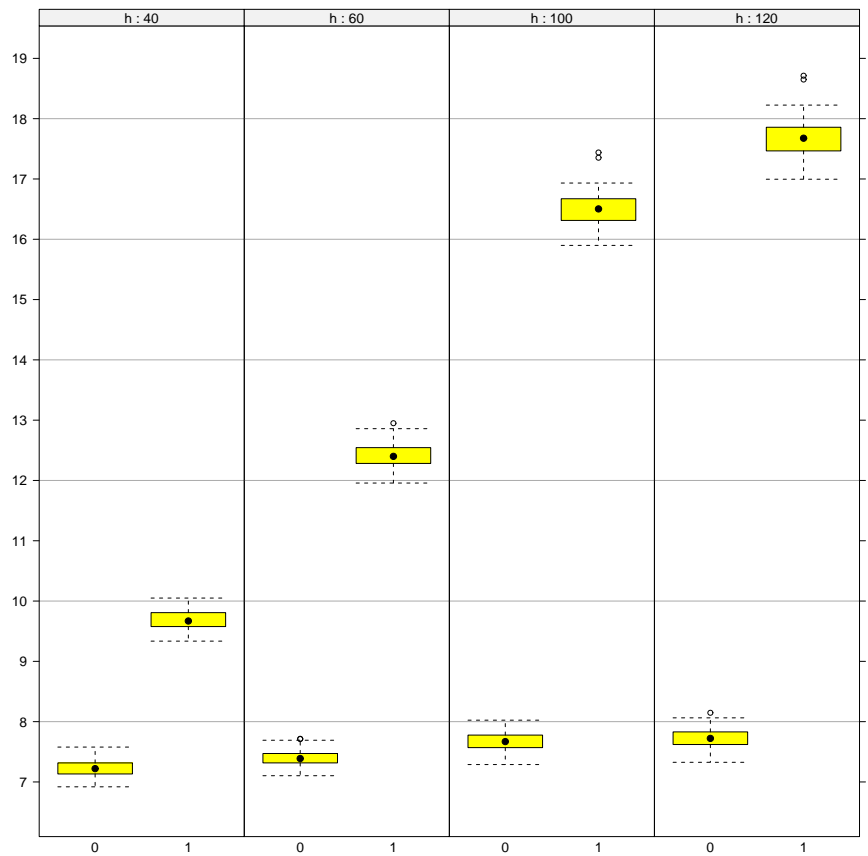}
\caption{ Boxplots of the estimates \( \widehat V_{0} \) (\( 0 \)) and \( \widetilde V_{0} \) (\( 1 \))
for different values of the bandwidth \( h \). }%
\label{Comparison_LBTR}%
\end{figure}
In Figure~\ref{Comparison_LBTR} we show the boxplots of \( \widetilde V_{0} \) and \( \widehat V_{0} \)
based on \( 100 \) sets of trajectories each of the size \( M=4000 \) (\( N=4000 \)) for different
values of the bandwidth \( h \), where
the triangle kernel \( K(x)=(1-\| x \|^{2})^{+} \) is used to construct \eqref{KE}.
The true value \( V_{0} \) of the option (computed
using a two-dimensional binomial lattice) is \( 8.08 \) in this case.
Several observations can be made by an examination of Figure~\ref{Comparison_LBTR}.
First, while the bias of \( \widehat V_{0} \) is always smaller then the bias
of \( \widetilde V_{0} \), the largest difference takes place for large \( h \).
This can be explained by the fact that for large \( h \) more observations
\( Y^{(m)}_{r+1} \) with \( X^{(m)}(t_{r}) \) lying far away from the given point
\( x \) become involved in the construction of \( \widehat C_{r,M}(x) \). This has
a consequence
of increasing the bias of the estimate \eqref{KE} and \( \widetilde V_{0} \) quickly deteriorates with increasing \( h \) . The most
interesting phenomenon is, however, the behavior
of \( \widehat V_{0} \) which turns out to be quite stable with respect to \( h \).
So, in the case of rather poor estimates of
continuation values (when \( h \) is increases) \( \widehat V_{0} \)  looks very
reasonable and even becomes closer to the true price.
\par
We stress that the aim of this example
is not to show the strength of the local polynomial estimation algorithms (although the performance
 of \( \widehat V_{0} \) for \( h=120 \) is quite comparable to the performance of a linear regression algorithm
 reported in \citet{Gl}) but rather to illustrate
the main message of this paper, namely the message about the efficiency of \( \widehat V_{0} \)
as compared to the estimates based on the direct use of continuation
values estimates.
\section{Conclusion}
In this paper we derive optimal rates of convergence for  low biased
estimates for the price of a Bermudan option based on suboptimal exercise policies
obtained from some estimates of
the optimal continuation values. We have shown that these rates are usually
much faster than the convergence rates
of the corresponding continuation values estimates. This may explain the
efficiency of these lower bounds observed in practice. Moreover, it turns out
that there are some cases where
the expected values of the lower bounds based on suboptimal stopping rules
achieve very fast convergence rates which are exponential in the number of paths
used to estimate the corresponding continuation values.

\section{Proofs}
\subsection{Proof of Proposition~\ref{CR}}
Define
\begin{eqnarray*}
\tau_{j}&:=&\min\{j\leq k < L: C_{k}(X(t_{k}))\leq f_{k}(X(t_{k}))\},\quad j=0,\ldots,L,
\\
\widehat\tau_{j,M}&:=&\min\{j\leq k < L: \widehat C_{k}(X(t_{k}))\leq f_{k}(X(t_{k}))\},\quad j=0,\ldots,L
\end{eqnarray*}%
and
\begin{eqnarray*}
V_{k,M}(x):=\E[f_{\widehat\tau_{k,M}}(X(t_{\widehat\tau_{k,M}}))|X(t_{k})=x], \quad x\in \mathbb{R}^{d}.
\end{eqnarray*}
The so called  \textit{Snell envelope} process \( V_{k} \) is related to \( \tau_{k} \) via
\begin{eqnarray*}
V_{k}(x)=\E[f_{\tau_{k}}(X(t_{\tau_{k}}))|X(t_{k})=x], \quad x\in \mathbb{R}^{d}.
\end{eqnarray*}
The following lemma provides a useful inequality which will be repeatedly
used in our analysis.
\begin{lem}
\label{BI}
For any \( k=0,\ldots,L-1 \), it holds with probability one
\begin{multline}
\label{BasicIneq}
0\leq V_{k}(X(t_{k}))-V_{k,M}(X(t_{k}))
\\
\leq \E^{\mathcal{F}_{t_{k}}}
\left[ \sum_{l=k}^{L-1}|f_{l}(X(t_{l}))-C_{l}(X(t_{l}))|
\right.
\\
\left.\times\left(\mathbf{1}_{\{\widehat\tau_{l,M}>l,\,\tau_{l}=l\}}
+\mathbf{1}_{\{\widehat\tau_{l,M}=l,\,\tau_{l}>l\}}\right)
\right].
\end{multline}

\end{lem}
\begin{proof}

We shall use  induction to prove \eqref{BasicIneq}. For \( k=L-1 \)
we have
\begin{multline}
\label{BasiqIneq2}
V_{L-1}(X(t_{L-1}))- V_{L-1,M}(X(t_{L-1}))=
\\
\nonumber
=
\E^{\mathcal{F}_{t_{L-1}}}\left[(f_{L-1}(X(t_{L-1}))-f_{L}(X(t_{L})))
\mathbf{1}_{\{\tau_{L-1}=L-1,\,\widehat\tau_{L-1,M}=L\}}\right]
\\
\nonumber
+\E^{\mathcal{F}_{t_{L-1}}}\left[(f_{L}(X(t_{L}))-f_{L-1}(X(t_{L-1})))
\mathbf{1}_{\{\tau_{L-1}=L,\,\widehat\tau_{L-1,M}=L-1\}}\right]
\\
\nonumber
=|f_{L-1}(X(t_{L-1}))-C_{L-1}(X(t_{L-1}))|\mathbf{1}_{\{\widehat\tau_{L-1,M}\neq\tau_{L-1}\}}
\end{multline}
since events \( \{ \tau_{L-1}=L \} \) and \( \{ \widehat\tau_{L-1,M}=L \} \) are measurable w.r.t. \( \mathcal{F}_{t_{L-1}} \).  Thus, \eqref{BasicIneq} holds with \( k=L-1 \). Suppose
that \eqref{BasicIneq} holds with
\( k=L'+1 \). Let us prove it for \( k=L' \). Consider a decomposition
\begin{eqnarray*}
f_{\tau_{L'}}(X(t_{\tau_{L'}}))-f_{\widehat\tau_{L',M}}(X(t_{\widehat\tau_{L',M}}))&=& S_{1}+S_{2}+S_{3}
\end{eqnarray*}
with
\begin{eqnarray*}
S_{1}&:=&\left( f_{\tau_{L'}}(X(t_{\tau_{L'}}))-f_{\widehat\tau_{L',M}}(X(t_{\widehat\tau_{L',M}})) \right)
\mathbf{1}_{\{\tau_{L'}>L',\,\widehat\tau_{L',M}>L'\}}
\\
S_{2}&:=&\left( f_{\tau_{L'}}(X(t_{\tau_{L'}}))-f_{\widehat\tau_{L',M}}(X(t_{\widehat\tau_{L',M}})) \right)
\mathbf{1}_{\{\tau_{L'}>L',\,\widehat\tau_{L',M}=L'\}}
\\
S_{3}&:=&\left( f_{\tau_{L'}}(X(t_{\tau_{L'}}))-f_{\widehat\tau_{L',M}}(X(t_{\widehat\tau_{L',M}})) \right)
\mathbf{1}_{\{\tau_{L'}=L',\,\widehat\tau_{L',M}>L'\}}.
\end{eqnarray*}
Since
\begin{eqnarray*}
\E^{\mathcal{F}_{t_{L'}}}\left[S_{1}\right]
&=&
\E^{\mathcal{F}_{t_{L'}}}\left[ \left( V_{L'+1}(X(t_{L'+1}))- V_{L'+1,M}(X(t_{L'+1})) \right)
\right]\mathbf{1}_{\{\tau_{L'}>L',\,\widehat\tau_{L',M}>L'\}},
\\
\E^{\mathcal{F}_{t_{L'}}}\left[S_{2}\right]&=&
\left( \E^{\mathcal{F}_{t_{L'}}}\left[ f_{\tau_{L'+1}}(X(t_{\tau_{L'+1}})) \right]
 -f_{L'}(X(t_{L'})) \right)
\mathbf{1}_{\{\tau_{L'}>L',\,\widehat\tau_{L',M}=L'\}}
\\
&=&\left(C_{L'}(X(t_{L'})) -f_{L'}(X(t_{L'})) \right)
\mathbf{1}_{\{\tau_{L'}>L',\,\widehat\tau_{L',M}=L'\}}
\end{eqnarray*}
and
\begin{eqnarray*}
\E^{\mathcal{F}_{t_{L'}}}\left[S_{3} \right]&=&
\left( f_{L'}(X(t_{L'}))-\E^{\mathcal{F}_{t_{L'}}}\left[ f_{\widehat\tau_{L'+1,M}}(X(t_{\widehat\tau_{L'+1,M}}))\right]  \right)
\mathbf{1}_{\{\tau_{L'}=L',\,\widehat\tau_{L',M}>L'\}}
\\
&=&\left(f_{L'}(X(t_{L'}))-C_{L'}(X(t_{L'}))\right)
\mathbf{1}_{\{\tau_{L'}=L',\,\widehat\tau_{L',M}>L'\}}
\\
&&+\E^{\mathcal{F}_{t_{L'}}}\left[ \left( V_{L'+1}(X(t_{L'+1}))- V_{L'+1,M}(X(t_{L'+1})) \right)
\mathbf{1}_{\{\tau_{L'}=L',\,\widehat\tau_{L',M}>L'\}}\right],
\end{eqnarray*}
we get with probability one
\begin{eqnarray*}
V_{L'}(X(t_{L'}))-V_{L',M}(X(t_{L'})&\leq&
\left|f_{L'}(X(t_{L'}))-C_{L'}(X(t_{L'}))\right|
\\
&&\times
\left(\mathbf{1}_{\{\widehat\tau_{L',M}>L',\,\tau_{L'}=L'\}}
+\mathbf{1}_{\{\widehat\tau_{L',M}=L',\,\tau_{L'}>L'\}}\right)
\\
&&+\E^{\mathcal{F}_{t_{L'}}}\left[V_{L'+1}(X(t_{L'+1}))- V_{L'+1,M}(X(t_{L'+1}))
\right].
\end{eqnarray*}
Our induction assumption implies now that
\begin{multline*}
V_{L'}(X(t_{L'}))-V_{L',M}(X(t_{L'}))\leq
\\
\E^{\mathcal{F}_{t_{L'}}}\left[ \sum_{l=L'}^{L-1}|f_{l}(X_{l})-C_{l}(X_{l})|
\left(\mathbf{1}_{\{\widehat\tau_{l,M}>l,\,\tau_{l}=l\}}
+\mathbf{1}_{\{\widehat\tau_{l,M}=l,\,\tau_{l}>l\}}\right)
\right]
\end{multline*}
and hence \eqref{BasicIneq} holds for \( k=L' \).
\end{proof}
Let us continue with the proof of Proposition~\ref{CR}.
Consider the sets \( \mathcal{E}_{l},\, \mathcal{A}_{l,j}\subset \mathbb{R}^{d},\,l=0,\ldots,L-1,\, j=1,2,\ldots, \)
defined as
\begin{eqnarray*}
\mathcal{E}_{l}&:=&\left\{ x\in \mathbb{R}^{d}: \widehat C_{l,M}(x)\leq f_{l}(x),\, C_{l}(x)>f_{l}(x)  \right\}
\\
&&\cup \left\{ x\in \mathbb{R}^{d}: \widehat C_{l,M}(x)>f_{l}(x),\, C_{l}(x)\leq f_{l}(x)  \right\},
\\
\mathcal{A}_{l,0}&:=&\left\{ x\in \mathbb{R}^{d}: 0<\left| C_{l}(x)-f_{l}(x) \right|\leq\gamma^{-1/2}_{M}\right\},
\\
\mathcal{A}_{l,j}&:=&\left\{ x\in \mathbb{R}^{d}: 2^{j-1}\gamma^{-1/2}_{M}<\left| C_{l}(x)-f_{l}(x) \right|\leq 2^{j}\gamma^{-1/2}_{M}
\right\}, \quad j>0.
\end{eqnarray*}
We may write
\begin{eqnarray*}
V_{0}(X(t_{0}))-V_{0,M}(X(t_{0}))
&\leq& \E^{\mathcal{F}_{t_{0}}}
\left[\sum_{l=0}^{L-1}|f_{l}(X(t_{l}))-C_{l}(X(t_{l}))|\mathbf{1}_{\{X(t_{l})\in \mathcal{E}_{l}\}}\right]
\\
&=&\sum_{j=0}^{\infty}\E^{\mathcal{F}_{t_{0}}}
\left[\sum_{l=0}^{L-1}|f_{l}(X(t_{l}))-C_{l}(X(t_{l}))|
\mathbf{1}_{\{X(t_{l})\in \mathcal{A}_{l,j}\cap \mathcal{E}_{l}\}}\right]
\\
&\leq& \gamma_{M}^{-1/2}\sum_{l=0}^{L-1}\P_{t_{l}|t_{0}}\left(0<\left| C_{l}(X(t_{l}))-f_{l}(X(t_{l})) \right|\leq\gamma_{M}^{-1/2}\right)
\\
&& +\sum_{j=1}^{\infty}\E^{\mathcal{F}_{t_{0}}}
\left[\sum_{l=0}^{L-1}|f_{l}(X(t_{l}))-C_{l}(X(t_{l}))|\mathbf{1}_{\{X(t_{l})\in \mathcal{A}_{l,j}\cap \mathcal{E}_{l}\}}\right].
\end{eqnarray*}
Using the fact that
\[
|f_{l}(X(t_{l}))-C_{l}(X(t_{l}))|\leq |\widehat C_{l,M}(X(t_{l})-C_{l}(X(t_{l}))|,
\quad l=0,\ldots,L-1,
\]
on \(\mathcal{E}_{l}\), we get for any \( j\geq 1 \) and \( l\geq 0 \)
\begin{multline*}
\E^{\mathcal{F}_{t_{0}}}\E_{\P_{x_{0}}^{\otimes M}}
\left[|f_{l}(X(t_{l}))-C_{l}(X(t_{l}))|\mathbf{1}_{\{X(t_{l})\in \mathcal{A}_{l,j}\cap \mathcal{E}_{l}\}}\right]
\\
\leq 2^{j}\gamma_{M}^{-1/2} \E^{\mathcal{F}_{t_{0}}}\E_{\P_{x_{0}}^{\otimes M}}
\left[\mathbf{1}_{\{|\widehat C_{l,M}(X(t_{l})-C_{l}(X(t_{l}))|\geq 2^{j-1}\gamma_{M}^{-1/2}\}}
\right.
\\
\left.
\times \mathbf{1}_{\{0<|f_{l}(X(t_{l}))-C_{l}(X(t_{l}))|\leq 2^{j}\gamma_{M}^{-1/2}\}}\right]
\\
\leq 2^{j}\gamma_{M}^{-1/2} \E^{\mathcal{F}_{t_{0}}}
\left[\P_{x_{0}}^{\otimes M}(|\widehat C_{l,M}(X(t_{l}))-C_{l}(X(t_{l}))|
\geq 2^{j-1}\gamma_{M}^{-1/2})\right.
\\
\left. \times \mathbf{1}_{\{0<|f_{l}(X(t_{l}))-C_{l}(X(t_{l}))|\leq 2^{j}\gamma_{M}^{-1/2}\}}\right]
\\
\leq B_{1}2^{j}\gamma_{M}^{-1/2}\exp\left( -B_{2} 2^{j-1} \right)\P_{t_{l}|t_{0}}(0<|f_{l}(X(t_{l}))-C_{l}(X(t_{l}))|\leq 2^{j}\gamma_{M}^{-1/2})
\\
\leq
B_{1} B_{0,l}2^{j(1+\alpha)}\gamma_{M}^{-(1+\alpha)/2}\exp\left(-B_{2}2^{j-1}\right),
\end{multline*}
where Assumption~\ref{BA} is used to get the last inequality. Finally, we get
\begin{multline*}
V_{0}(X(t_{0}))-\E_{\P_{x_{0}}^{\otimes M}}\left[ V_{0,M}(X(t_{0})) \right]
\\
\leq \left[ \sum_{l=0}^{L-1}B_{0,l} \right]\gamma_{M}^{-(1+\alpha)/2}
+ B'\left[\sum_{l=0}^{L-1}B_{0,l} \right]\gamma_{M}^{-(1+\alpha)/2}\sum_{j\geq 1} 2^{j(1+\alpha)}\exp(-B_{2}2^{j-1})
\\
\leq B\left[\sum_{l=0}^{L-1}B_{0,l} \right]\gamma_{M}^{-(1+\alpha)/2}
\end{multline*}
with some constant \( B \) depending on \( B_{1} \), \( B_{2} \) and \( \alpha. \)
\subsection{Proof of Proposition~\ref{LowerBound}}
We have
\begin{multline}
\label{BasiqIneqLB}
V_{0}(X(t_{0}))-\widehat V_{0,M}(X(t_{0}))=
\\
=
\E^{\mathcal{F}_{t_{0}}}\left[(f_{1}(X(t_{1}))-f_{2}(X(t_{2})))1(\tau_{1}=1,\widehat\tau_{1,M}=2)\right]
\\
+\E^{\mathcal{F}_{t_{0}}}\left[(f_{2}(X(t_{2}))-f_{1}(X(t_{1})))1(\tau_{1}=2,\widehat\tau_{1,M}=1)\right]
\\
=\E^{\mathcal{F}_{t_{0}}}\left[ |f_{1}(X(t_{1}))-C_{1}(X(t_{1}))|
\mathbf{1}_{\{\widehat\tau_{1,M}\neq\tau_{1}\}} \right].
\end{multline}
\\
For an integer \( q\geq 1 \) consider a regular grid on \( [0,1]^{d} \)
defined as
\begin{eqnarray*}
    G_{q}=\left\{ \left( \frac{2k_{1}+1}{2q},\ldots, \frac{2k_{d}+1}{2q}\right):\, k_{i}\in \{ 0,\ldots,q-1 \},\, i=1,\ldots,d \right\}.
\end{eqnarray*}
Let \( n_{q}(x)\in G_{q} \) be the closest point to \( x\in \mathbb{R}^{d} \) among
points in \( G_{q} \). Consider the partition \( \mathcal{X}'_{1},\ldots,\mathcal{X}'_{q^{d}} \)
of \( [0,1]^{d} \) canonically defined using the grid \( G_{q} \) (\( x\) and \( y \) belong to the same subset if and only if \( n_{q}(x)=n_{q}(y) \)). Fix an integer \( m\leq q^{d} \). For any
\( i\in \{ 1,\ldots,m \} \), define \( \mathcal{X}_{i}=\mathcal{X}'_{i} \) and \( \mathcal{X}_{0}=\mathbb{R}^{d}\setminus \bigcup_{i=1}^{m}\mathcal{X}_{i} \), so that \( \mathcal{X}_{0},\ldots,\mathcal{X}_{m} \) form a partition of \( \mathbb{R}^{d} \).
Denote by \( \mathcal{B}_{q,j} \) the ball with the center in \( n_{q}(\mathcal{X}_{j}) \)
and radius \( 1/2q \).
\par
 Define a hypercube \( \mathcal{H}=\{ P_{\bar \sigma}: \bar\sigma=(\sigma_{1},\ldots,\sigma_{m}) \in \{ -1,1 \}^{m}\} \) of probability distributions \( \P_{\bar \sigma} \) of the r.v. \( (X(t_{1}),f_{2}(X(t_{2}))) \) valued in \( \mathbb{R}^{d}\times \{ 0,1 \} \) as follows.
For any \( \P_{\bar \sigma}\in \mathcal{H} \) the marginal distribution of \( X(t_{1}) \)
(given \( X(t_{0})=x_{0} \)) does not depend on \( \bar\sigma \) and has a bounded density \( \mu \) w.r.t.  the Lebesgue measure
on \( \mathbb{R}^{d} \) such that \( \P_{\mu}(\mathcal{X}_{0})=0 \) and
\[
\P_{\mu}(\mathcal{X}_{j})=\P_{\mu}(\mathcal{B}_{q,j})=\int_{\mathcal{B}_{q,j}}\mu(x)\,dx=\omega,
\quad j=1,\ldots,m
\]
for some \( \omega>0 \).
In order to ensure that the density \( \mu \) remains  bounded
we assume that \( q^{d}\omega =O(1) \).
\par
The distribution of \( f_{2}(X(t_{2})) \) given \( X(t_{1}) \) is determined by
the probability \( \P_{\bar\sigma}(f_{2}(X(t_{2}))=1|X(t_{1})=x) \) which is equal to
\( C_{1,\bar \sigma}(x) \).  Define
\[
 C_{1,\bar\sigma}(x)=f_{1}(x)+\sigma_{j}\phi(x), \quad x\in \mathcal{X}_{j}, \quad
 j=1,\ldots,m,
\]
and \( C_{1,\bar\sigma}(x)=f_{1}(x) \) on \( \mathcal{X}_{0} \), where
\( \phi(x)=\gamma^{-1/2}_{M}\varphi(q[x-n_{q}(x)]) \), \( \varphi(x)=A_{\varphi}\theta(\| x \|) \)
with some constant \( A_{\varphi}>0 \) and with
\( \theta:\mathbb{R}_{+}\to \mathbb{R}_{+} \) being a non-increasing
infinitely differentiable function such that \( \theta(x)\equiv 1 \) on
\( [0,1/2] \) and \( \theta(x)\equiv 0 \) on \( [1,\infty) \). Furthermore, there exist two real numbers \( 0<f_{-}<f_{+}<1 \) such that
\( f_{-}\leq f_{1}(x) \leq f_{+}. \)
Taking
\( A_{\varphi} \) small enough, we can then ensure that
\( 0\leq C_{1,\bar \sigma}(x)\leq 1 \)
on \( \mathbb{R}^{d} \).
Obviously, it holds \( \phi(x)=A_{\varphi}\gamma^{-1/2}_{M} \)
for \( x\in \mathcal{B}_{q,j} \).
As to the boundary assumption \eqref{BA}, we have
\begin{multline*}
\P_{\mu}(0<|f_{1}(X(t_{1}))-C_{1,\bar \sigma}(X(t_{1}))|\leq\delta)=
\\
\sum_{j=1}^{m}\P_{\mu}(0<|f_{1}(X(t_{1}))-C_{1,\bar \sigma}(X(t_{1}))|
\leq\delta, X(t_{1})\in \mathcal{B}_{q,j})
\\
 =\sum_{j=1}^{m}\int_{\mathcal{B}_{q,j}}\mathbf{1}_{\{0<\phi(x)\leq\delta\}}\mu(x)\,dx
 =m\omega\mathbf{1}_{\{A_{\varphi}\gamma_{M}^{-1/2}\leq\delta\}}
\end{multline*}
and \eqref{BA} holds provided that \( m\omega=O(\gamma^{-\alpha/2}_{M}) \).
Let \( \widehat\tau_{M} \) be a stopping time measurable w.r.t. \( \mathcal{F}^{\otimes M} \),
then the identity \eqref{BasiqIneqLB} leads  to
\begin{multline*}
 \E^{\mathcal{F}_{t_{0}}}_{\P_{\bar\sigma}}[f_{\tau}(X(\tau))]-
 \E_{\P_{\bar\sigma}^{\otimes M}}[\E^{\mathcal{F}_{t_{0}}}f_{\widehat\tau_{M}}(X(\widehat\tau_{M}))]
\\
 =\E_{\P_{\bar\sigma}^{\otimes M}}\E_{P_\mu}^{\mathcal{F}_{t_{0}}}
 \left[ |\Delta_{\bar\sigma}(X(t_{1}))|\mathbf{1}_{\{\widehat\tau_{1,M}\neq\tau_{1}\}}
 \right],
\end{multline*}
with \( \Delta_{\bar\sigma}(X(t_{1}))= f_{1}(X(t_{1}))-C_{1,\bar\sigma}(X(t_{1}))\).
By conditioning on \( X(t_{1}), \) we get
\begin{multline*}
    \E_{\P_{\bar \sigma}^{\otimes M}}
    \E_{\P_{\mu}}^{\mathcal{F}_{t_{0}}}\left[ |\Delta_{\bar \sigma}(X(t_{1}))|
    \mathbf{1}_{\{\widehat\tau_{1,M}\neq\tau_{1}\}} \right]
    \\
    =\omega \sum_{j=1}^{m}\E_{\P_{\bar \sigma}^{\otimes M}}\E_{\P_{\mu}}^{\mathcal{F}_{t_{0}}}
    \left[\phi(X(t_{1}))\mathbf{1}_{\{\widehat\tau_{1,M}\neq\tau_{1}\}}|X(t_{1})\in \mathcal{B}_{q,j} \right]
    \\
     =A_{\varphi}m\omega \gamma_{M}^{-1/2}\E_{\P_{\mu}}^{\mathcal{F}_{t_{0}}}\P_{\bar \sigma}^{\otimes
     M}(\widehat\tau_{1,M}\neq\tau_{1}).
\end{multline*}
Using now a well known Birg\'e's or Huber's lemma \citep[see, e.g.][p. 243]{DGL}, we get
\begin{eqnarray*}
\sup_{\bar\sigma\in \{ -1;+1 \}^{m}}\P_{\bar \sigma}^{\otimes
     M}(\widehat\tau_{1,M}\neq\tau_{1})\geq  \left[ 0.36 \wedge \left( 1-\frac{M K_{\mathcal{H}}}{\log(\mathcal{|H|})} \right)
     \right],
\end{eqnarray*}
where
\( K_{\mathcal{H}}:=\sup_{P,Q\in \mathcal{H}}K(P,Q) \) and \( K(P,Q) \) is
a Kullback-Leibler distance between two measures \( P \) and \( Q \).
Since for any two measures \( P \) and \( Q \) from \( \mathcal{H} \) with \( Q\neq P \)
it holds
\begin{eqnarray*}
K(P,Q)&\leq &\sup_{\substack{
\bar\sigma_{1},\bar\sigma_{2}\in \{ -1;+1 \}^{m} \\ \bar\sigma_{1}\neq \bar\sigma_{2}} }
\E_{\P_{\mu}}^{\mathcal{F}_{t_{0}}}\left[ C_{1,\bar \sigma_{2}}(X(t_{1}))
\log \left\{ \frac{C_{1,\bar \sigma_{1}}(X(t_{1}))}{C_{1,\bar \sigma_{2}}(X(t_{1}))} \right\}
\right.
\\
&&\left.+(1-C_{1,\bar \sigma_{2}}(X(t_{1})))
\log \left\{ \frac{1-C_{1,\bar \sigma_{1}}(X(t_{1}))}{1-C_{1,\bar \sigma_{2}}(X(t_{1}))} \right\} \right]
\\
&\leq&  (1-f_{+}-A_{\varphi})^{-1}(f_{-}-A_{\varphi})^{-1}\E_{\P_{\mu}}^{\mathcal{F}_{t_{0}}}
\left[ \phi^{2}(X(t_{1}))
\mathbf{1}_{\{X(t_{1})\not\in \mathcal{X}_{0}\}} \right]
\end{eqnarray*}
for small enough \( A_{\varphi} \), and \( \log(|\mathcal{H}|)=m \log (2) \), we get
\begin{multline*}
\sup_{\bar\sigma\in \{ -1;+1 \}^{m}}\left\{\E^{\mathcal{F}_{t_{0}}}_{\P_{\bar\sigma}}[f_{\tau,\bar\sigma}(X(\tau))]-\E_{\P_{\bar\sigma}^{\otimes M}}[\E^{\mathcal{F}_{t_{0}}}f_{\widehat\tau_{M},\bar\sigma}(X(\widehat\tau_{M}))]
\right\}\geq
\\
A_{\varphi}m\omega \gamma^{-1/2}_{M}(1-AM\gamma^{-1}_{M}\omega)\gtrsim
\gamma_{M}^{-(1+\alpha)/2},
\end{multline*}
provided that \( m\omega>B\gamma^{-\alpha/2}_{M} \) for some \( B>0 \) and
\( AM\omega<\gamma_{M}, \) where \( A \) is a positive
constant depending on \( f_{-}, f_{+} \) and \( A_{\varphi}. \) Using similar
arguments, we  derive
\begin{eqnarray*}
\sup_{\bar\sigma\in \{ -1;+1 \}^{m}}\P_{\bar \sigma}^{\otimes M}(|C_{1,\bar \sigma}(x)-
\widehat C_{1,M}(x)|>\delta\gamma_{M}^{-1/2})>0
\end{eqnarray*}
for almost \( x \) w.r.t. \( \P_{\mu} \), some \( \delta>0 \) and any estimator \( \widehat C_{1,M} \) measurable w.r.t. \( \mathcal{F}^{\otimes M} \).
\subsection{Proof of Proposition~\ref{CRF}}
Using the arguments similar to ones in the proof of Proposition~\ref{CR}, we get
\begin{multline}
\label{CRFE1}
V_{0}(X(t_{0}))-\E_{\P_{x_{0}}^{\otimes M}}\left[V_{0,M}(X(t_{0}))\right]
\leq
\\
\delta_{0}\sum_{l=0}^{L-1}
\P_{t_{l}|t_{0}}(0<| C_{l}(X(t_{l}))-f_{l}(X(t_{l}))|\leq\delta_{0})
\\
+ \sum_{l=0}^{L-1}\E^{\mathcal{F}_{t_{0}}}
\E_{\P_{x_{0}}^{\otimes M}}\left[| C_{l}(X(t_{l}))-f_{l}(X(t_{l}))|
\right.
\\
\left.
\times\mathbf{1}_{\{X(t_{l})\in \mathcal{E}_{l}\}}
\mathbf{1}_{\{| C_{l}(X(t_{l}))-f_{l}(X(t_{l}))|>\delta_{0}\}}\right]
\end{multline}
with \( \mathcal{E}_{l} \) defined  as in the proof of Proposition~\ref{CR}.
The first summand on the right-hand side of \eqref{CRFE1} is equal to zero due
to \eqref{BAF}. Hence, Cauchy-Schwarz and Minkowski inequalities imply
\begin{eqnarray*}
V_{0}(X(t_{0}))-\E_{\P_{x_{0}}^{\otimes M}}\left[V_{0,M}(X(t_{0}))\right]&\leq &  \sum_{l=0}^{L-1}\left[\E^{\mathcal{F}_{t_{0}}}
| \E^{\mathcal{F}_{t_{l}}} \left[ f_{\tau_{l+1}}(X(t_{\tau_{l+1}})) \right]-f_{l}(X(t_{l}))|^{2}\right]^{1/2}
\\
&&\times\left[\E^{\mathcal{F}_{t_{0}}}\P_{x_{0}}^{\otimes M}(| C_{l}(X(t_{l}))-\widehat C_{l,M}(X(t_{l}))|>\delta_{0}) \right]^{1/2}
\\
&\leq &  2B^{1/2}_{f}\sum_{l=0}^{L-1}\left[ \E^{\mathcal{F}_{t_{0}}}\P_{x_{0}}^{\otimes M}(| C_{l}(X(t_{l}))-\widehat C_{l,M}(X(t_{l}))|>\delta_{0})
\right]^{1/2}.
\end{eqnarray*}
Now the application of \eqref{ExpBF} finishes the proof.
\subsection{Proof of Proposition~\ref{ExpBoundsKern}}
Denote
\begin{eqnarray*}
    \varepsilon_{k,M}(x)=T[\widehat C_{k,M}](x)-C_{k}(x)
\end{eqnarray*}
and
\begin{eqnarray*}
    \zeta_{k,M}(x)=\widetilde C_{k,M}(x)-T[\widehat C_{k,M}](x)
\end{eqnarray*}
for \( k=1,\ldots, L-1 \).
Using the elementary inequality \( |\max(a,x)-\max(a,y)|\leq |x-y|, \) which holds
for any real  numbers \( a \), \( x \) and \( y \), we get
\begin{eqnarray*}
  |\varepsilon_{k,M}(x)|\leq |\zeta_{k,M}(x)|+ \E\left[\left.|\varepsilon_{k+1,M}(X(t_{k+1}))|\right|X(t_{k})=x  \right]
\end{eqnarray*}
and hence
\begin{eqnarray}
\label{EBP}
    |\varepsilon_{k,M}(x)|&\leq &
    \sum_{l=k+1}^{L-1} \E\left[|\zeta_{l,M}(X(t_{l}))||X(t_{k})=x  \right]
    \\
    \nonumber
    &:=&\sum_{l=k+1}^{L-1} \xi_{l,k,M}(x).
\end{eqnarray}
Note that we take expectation in \eqref{EBP} with respect to a new \( \sigma \)-algebra
\( \mathcal{F} \) which is independent of \( \mathcal{F}^{\otimes M} \) and
 \( \{ \zeta_{l,M} \} \) are measurable w.r.t \( \mathcal{F}^{\otimes M} \). Hence, random variables \( \{ \xi_{l,k,M} \} \) are \( \mathcal{F}^{\otimes M}  \)
measurable as well.  According to Lemma~\ref{1PBounds} (see below)
\begin{multline*}
    \P^{\otimes M}_{x_{0}}\left(\xi_{l,k,M}(x)\geq \delta\sqrt{|\log h|/Mh^{d+2\nu }}\right)
   \leq
   \\
   \P_{x_{0}}^{\otimes M}\left(\sup_{y\in \mathcal{A}}|\zeta_{l,M}(y)|\geq\delta\sqrt{|\log h|/Mh^{d+2\nu }}\right)
   \leq D_{2}\exp(-D_{3}\delta)
\end{multline*}
for almost all \( x \) w.r.t. \( \P_{t_{k}|t_{0}} \).
Thus,
\begin{eqnarray*}
    \P_{x_{0}}^{\otimes M}\left(|\varepsilon_{k,M}(x)|\geq \delta\sqrt{|\log h|/Mh^{d+2\nu }}\right)
    \leq LD_{2}\exp(-D_{3}\delta/L).
\end{eqnarray*}
Analogously, using Lemma~\ref{2PBounds} one can prove that
\begin{eqnarray*}
    \P_{x_{0}}^{\otimes M}\left(|\varepsilon_{k,M}(x)|\geq \delta\right)
    \leq B_{4}\exp(-B_{5}Mh^{d+\nu })
\end{eqnarray*}
with some positive constants \( B_{4} \) and \( B_{5}. \)

\begin{lem}
\label{1PBounds}
Let assumptions (AX0)-(AX2), (AK1) and (AK2) be fulfilled.
Then there exist positive constants \( D_{1} \), \( D_{2} \) and  \( D_{3}\),  such that
for any \( h \) satisfying  \( D_{1}h^{\beta}<\sqrt{|\log h|/Mh^{d}}  \)
the estimates
\( \{ T[\widehat C_{k,M}] \} \) based on the truncated local polynomials
estimators of degree \( \lfloor \beta \rfloor \) fulfill
\begin{eqnarray*}
    \P_{x_{0}}^{\otimes M}\left(\sup_{x\in \mathcal{A}}|T[\widehat C_{k,M}](x)-\widetilde C_{k}(x)|\geq \delta\sqrt{|\log h|/Mh^{d+2\nu }}\right)
    \leq D_{2}\exp(-D_{3}\delta),
\end{eqnarray*}
for all \( \delta>\delta_{0} \) and  \( k=1,\ldots,L-1 \).
\end{lem}
\begin{lem}
\label{2PBounds}
Let assumptions (AX0)-(AX2), (AK1) and (AK2) be fulfilled and \( \sqrt{|\log h|/Mh^{d+2\nu }}=o(1) \)
for \( M\to \infty. \)
Then  there exist positive constants  \( D_{4} \), \( D_{5} \)
and \( D_{6} \) such that for any \( \delta\geq D_{4}h^{\beta} \) the inequality
\begin{eqnarray*}
    \P_{x_{0}}^{\otimes M}\left(\sup_{x\in \mathcal{A}}|T[\widehat C_{k,M}](x)-\widetilde C_{k}(x)|\geq \delta\right)
    \leq D_{5}\exp(-D_{6}Mh^{d+\nu })
\end{eqnarray*}
holds for all \( k=1,\ldots,L-1 \).
\end{lem}
\begin{proof} We give the proof only for Lemma~\ref{1PBounds}. Lemma~\ref{2PBounds}
can be proved in a similar way.
Fix some natural \( r>0 \) such that \( 0<r\leq L \)
and consider  the matrix
\( \Gamma=(\Gamma_{u_{1},u_{2}})_{|u_{1}|,|u_{2}|\leq \lfloor \beta \rfloor} \) with elements
\begin{eqnarray*}
  \Gamma_{u_{1},u_{2}}=\frac{1}{Mh^{d}}\sum_{m=1}^{M}\left(\frac{X^{(m)}(t_{r})-x}
  {h}\right)^{u_{1}+u_{2}}K\left( \frac{X^{(m)}(t_{r})-x}{h} \right).
\end{eqnarray*}
The smallest eigenvalue \( \lambda_{\Gamma} \) of the matrix \( \Gamma \) satisfies
\begin{eqnarray}
\label{lamdag}
\nonumber
\lambda_{\Gamma}&=&\min_{\| W \|=1}W^{\top}\Gamma W
\\
\nonumber
&\geq& \min_{\| W \|=1}W^{\top}\E[\Gamma] W+\min_{\| W \|=1}W^{\top}(\Gamma-\E [\Gamma]) W
\\
&\geq& \min_{\| W \|=1}W^{\top}\E[\Gamma] W-\sum_{|u_{1}|,|u_{2}|\leq
\lfloor\beta\rfloor}|\Gamma_{u_{1},u_{2}}-\E[\Gamma_{u_{1},u_{2}}]|.
\end{eqnarray}
By Assumption (AX2)
\begin{eqnarray*}
\inf_{x\in \mathcal{A}}
\min_{\| W \|=1}\left[ W^{\top}\E[\Gamma(x)] W \right]\geq\gamma_{0}h^{\nu }
\end{eqnarray*}
with some \( \gamma_{0}>0. \)
For \( m=1,\ldots, M, \) and any multi-indices \( u_{1} \), \( u_{2} \)
such that \( |u_{1}|, |u_{2}|\leq \lfloor\beta\rfloor \), define
\begin{multline*}
\Delta_{m}(x)=\frac{1}{h^{d}}\left( \frac{X^{(m)}(t_{r})-x}{h} \right)^{u_{1}+u_{2}}
K\left( \frac{X^{(m)}(t_{r})-x}{h} \right)
\\
-\int_{\mathbb{R}^{d}}z^{u_{1}+u_{2}}
K(z)p(t_{r},x+hz|t_{0},x_{0})\, dz.
\end{multline*}
We have \( \E_{\P_{t_{r}|t_{0}}} [\Delta_{m}(x)]=0 \),
\[
|\Delta_{m}(x)|\leq h^{-d}\sup_{z\in \mathbb{R}^{d}}\left[ (1+\| z \|^{2\beta})K(z)
\right]=:K_{1}h^{-d}
\]
and
\begin{eqnarray*}
\E_{\P_{t_{r}|t_{0}}} [\Delta_{m}(x)]^{2}&\leq & \int_{\mathbb{R}^{d}}z^{2u_{1}+2u_{2}}
K^{2}(z)p(t_{r},x+hz|t_{0},x_{0})\, dz
\\
&\leq & \frac{p_{\max}}{h^{d}}\int_{\mathbb{R}^{d}}(1+\| z \|^{4\beta})K^{2}(z)\,dz=:
K_{2}h^{-d},
\end{eqnarray*}
where \( p_{\max}=\sup_{z\in \mathbb{R}^{d}} p(t_{r},z|t_{0},x_{0})\)
and \( K_{1}, K_{2} \) are two positive constants.
Due to assumption (AK2), the class of functions
\begin{eqnarray*}
\left\{ \left( \frac{x-\cdot}{h}\right)^{u_{1}+u_{2}} K\left( \frac{x-\cdot}{h}\right):\, x\in \mathbb{R}^{d},
\,h\in \mathbb{R}\setminus \{ 0 \}, \, |u_{1}|,|u_{2}|\leq \lfloor\beta\rfloor  \right\}
\end{eqnarray*}
is a bounded Vapnik-\v{C}ervonenkis  class of measurable functions (see \citet{Du}). According to
Proposition~\ref{EIP} (see Appendix), we have for any \( \zeta>0 \)
\begin{multline}
\label{EIGamma}
\P_{t_{r}|t_{0}}\left( \sup_{x\in \mathcal{A}}|\Gamma_{u_{1},u_{2}}(x)-\E\Gamma_{u_{1},u_{2}}(x)|
\geq \zeta  \right)
\\
=\P_{t_{r}|t_{0}}\left( \sup_{x\in \mathcal{A}}\frac{1}{M}
\left| \sum_{m=1}^{M}\Delta_{m}(x) \right|\geq\zeta \right)
\\
\leq L_{0}\exp(-\zeta B_{0} M h^{d})
\end{multline}
with some positive constants \( L_{0} \) and \( B_{0} \).
Combining  \eqref{lamdag} and \eqref{EV_WGW} with \eqref{EIGamma}, we get
\begin{eqnarray*}
    \P_{t_{r}|t_{0}}\left(\inf_{x\in \mathcal{A}}\lambda_{\Gamma}(x)
    \leq\gamma_{0}h^{\nu }/2\right)\leq  L_{0}N^{2}_{\beta}\exp(-\gamma_{0} B_{0} M
    h^{d+\nu }/2N^{2}_{\beta}),
\end{eqnarray*}
where \( N^{2}_{\beta} \) is the number of elements in the matrix \( \Gamma. \)
Assume that \( M \) is large enough so that \( \gamma_{0}/2>(\log M)^{-1}. \)
Then on the set \( \{ \inf_{x\in \mathcal{A}}\lambda_{\Gamma}(x)>\gamma_{0}h^{\nu }/2 \} \) we have
\[
|T[\widehat C_{r,M}](x)-\widetilde C_{r}(x)|\leq
|\widehat C_{r,M}(x)-\widetilde C_{r}(x)|, \quad x\in \mathcal{A}
\]
since \( \sup_{x\in \mathcal{A}}\widetilde C_{r}(x)\leq C_{\max}. \)
Therefore, it holds for any \( \zeta>0 \)
\begin{multline*}
    \P_{t_{r}|t_{0}}\left(\sup_{x\in \mathcal{A}}|T[\widehat C_{r,M}](x)-
    \widetilde C_{r}(x)|\geq \zeta\right)\leq
    \P_{t_{r}|t_{0}}\left(\inf_{x\in \mathcal{A}}\lambda_{\Gamma}(x)
     \leq\gamma_{0}h^{\nu }/2\right)
    \\
    +\P_{t_{r}|t_{0}}\left(\sup_{x\in \mathcal{A}}|\widehat C_{r,M}(x)
    -\widetilde C_{r}(x)|\geq \zeta,\,
    \inf_{x\in \mathcal{A}}\lambda_{\Gamma}(x)>\gamma_{0}h^{\nu }/2\right).
\end{multline*}
Introduce the matrix \( Q=(Q_{m,u})_{1\leq m \leq M,\, |u|\leq \lfloor\beta\rfloor} \) with elements
\begin{eqnarray*}
    Q_{m,u}=\left(\frac{X^{(m)}(t_{r})-x}{h}\right)^{u}\sqrt{\frac{1}{Mh^{d}}K\left( \frac{X^{(m)}(t_{r})-x}{h} \right)}.
\end{eqnarray*}
Denote by \( Q_{u} \) the \( u \)th column of \( Q \) and
define
\begin{eqnarray*}
    Q^{C}(x):=\sum_{|u|\leq \lfloor\beta\rfloor} \frac{\widetilde C_{r}^{(u)}(x)h^{u}}{u!}Q_{u}.
\end{eqnarray*}
Since \(\Gamma=Q^{\top}Q \), we get \( Z^{\top}(0)
\Gamma^{-1}Q^{\top}Q_{u}=
\mathbf{1}_{\{ u=(0,\ldots,0) \}}\) for any
\( s \) with \( |s|\leq \lfloor\beta\rfloor. \) Hence \( Z^{\top}(0)
\Gamma^{-1}Q^{\top}Q^{C}=\widetilde C_{r}(x) \).
Thus, we can write
\begin{eqnarray*}
    \widehat C_{r,M}(x)-\widetilde C_{r}(x)=
    Z^{\top}(0)\Gamma^{-1}(S-Q^{\top}Q^{C})=:Z^{\top}(0)\Gamma^{-1}\varepsilon_{M}(x),
\end{eqnarray*}
where \( \varepsilon_{M}(x) \) is a vector valued function  with components
\begin{eqnarray*}
    \varepsilon_{M,u}(x)&=&\frac{1}{Mh^{d}}\sum_{m=1}^{M}\left[Y_{r+1}^{(m)}-\widetilde C_{r,x}(X^{(m)}(t_{r}))\right]
    \left( \frac{X_{r}^{(m)}-x}{h}\right)^{u}K\left( \frac{X_{r}^{(m)}-x}{h} \right)
\end{eqnarray*}
and
\( Y_{r+1}^{(m)}=\max(f_{r+1}(X^{(m)}(t_{r+1})),T[\widehat C_{r+1,M}](X^{(m)}(t_{r+1}))). \)
So, on the set \( \{ \inf_{x\in \mathcal{A}}\lambda_{\Gamma}(x)>\gamma_{0}h^{\nu}/2 \} \) we get
\begin{eqnarray*}
  |\widehat C_{r,M}(x)-\widetilde C_{r}(x)|\leq \| \Gamma \varepsilon_{M} \|\leq
  \lambda_{\Gamma}^{-1}\| \varepsilon_{M} \|\leq 2h^{-\nu }\gamma_{0}^{-1}\| \varepsilon_{M} \|\leq
  2h^{-\nu }\gamma_{0}^{-1}N^{1/2}_{\beta}\max_{u}|\varepsilon_{M,u}(x)|.
\end{eqnarray*}
Denote
\begin{eqnarray*}
    \Delta^{(1)}_{u,m}(x)&:=&\frac{1}{h^{d}}\left[Y_{r+1}^{(m)}-\widetilde C_{r}(X^{(m)}(t_{r}))\right]\left( \frac{X_{r}^{(m)}-x}{h}\right)^{u}K\left( \frac{X_{r}^{(m)}-x}{h} \right),
    \\
    \Delta^{(2)}_{u,m}(x)&:=&\frac{1}{h^{d}}\left[\widetilde C_{r}(X^{(m)}(t_{r}))-\widetilde C_{r,x}(X^{(m)}(t_{r}))\right]\left( \frac{X_{r}^{(m)}-x}{h}\right)^{u}K\left( \frac{X_{r}^{(m)}-x}{h} \right).
\end{eqnarray*}
It holds
\begin{eqnarray*}
    |\varepsilon_{M,u}|\leq \left| \frac{1}{M}\sum_{m=1}^{M}\Delta_{u,m}^{(1)} \right|+\left| \frac{1}{M}\sum_{m=1}^{M}\left[ \Delta_{u,m}^{(2)}-\E\Delta_{u,m}^{(2)} \right] \right|+|\E\Delta_{u,m}^{(2)}|.
\end{eqnarray*}
Note that \( \E_{\P_{t_{r}|t_{0}}}  \left[ \Delta_{u,m}^{(1)} \right]=0\) and
\begin{eqnarray*}
    |\Delta_{u,m}^{(1)}(x)|\leq A_{11}h^{-d},\quad \Var \left[ \Delta_{u,m}^{(1)}(x) \right]\leq A_{12}h^{-d},
    \\
    \left| \Delta_{u,m}^{(2)}(x)-\E\left[ \Delta_{u,m}^{(2)}(x) \right] \right|\leq A_{21}h^{\beta-d},
    \quad \Var \left[ \Delta_{u,m}^{(2)}(x) \right]\leq A_{22}h^{2\beta-d}
\end{eqnarray*}
with some positive constants \( A_{11} \), \( A_{12} \), \( A_{21} \) and \( A_{22} \) not depending
on \( x \).
Proposition~\ref{EIP}
implies that for any \( \delta\geq \delta_{0}>0 \)
\begin{eqnarray*}
    \P_{t_{r}|t_{0}}\left( \left\| \frac{1}{M}\sum_{m=1}^{M}\Delta_{u,m}^{(1)} \right\|_{\infty} \geq\delta\sqrt{|\log h|/Mh^{d}}\right)\leq L_{1}\exp\left( -\delta B_{1} |\log h| \right)
\end{eqnarray*}
with some positive constants \( L_{1} \) and \( B_{1} \).
Furthermore, due to the representation
\begin{multline*}
\widetilde C_{r}(z)-\widetilde C_{r,x}(z)=\lfloor\beta\rfloor\sum_{|u|=\lfloor\beta\rfloor}
\frac{(z-x)^{u}}{u!}
\\
\times\int_{0}^{1}\left[ \widetilde C^{(u)}_{r}(x+w(z-x))
-\widetilde C^{(u)}_{r}(x) \right](1-w)^{\lfloor\beta\rfloor-1}\, dw
\end{multline*}
we get for any two points \( x_{1}\) and  \(x_{2}\) in \(\mathbb{R}^{d} \)
\begin{eqnarray*}
\|\widetilde C_{r}(\cdot)-\widetilde C_{r,x_{1}}(\cdot)-
(\widetilde C_{r}(\cdot)-\widetilde C_{r,x_{2}}(\cdot))\|_{\mathcal{A}}
\leq \| x_{1}-x_{2} \|^{\beta-\lfloor\beta\rfloor}.
\end{eqnarray*}
Now it can be shown (see \citet{Du}) that the class
\begin{eqnarray*}
\left\{ \left[\widetilde C_{r}(\cdot)-\widetilde C_{r,x}(\cdot)\right]\left( \frac{\cdot-x}{h}\right)^{u}K\left( \frac{\cdot-x}{h} \right):\, x\in \mathbb{R}^{d},
\,h\in \mathbb{R}\setminus \{ 0 \}, \, |u|\leq \lfloor\beta\rfloor  \right\}
\end{eqnarray*}
is a bounded Vapnik-\v{C}ervonenkis  class of measurable functions.
Hence
\begin{eqnarray*}
    \P_{t_{r}|t_{0}}\left( \left\| \frac{1}{M}\sum_{m=1}^{M}\left[ \Delta_{u,m}^{(2)}
    -\E_{\P_{t_{r}|t_{0}}} \Delta_{u,m}^{(2)} \right] \right\|_{\infty} \geq\delta\sqrt{|\log h|/Mh^{d}}\right)
    \leq L_{2}\exp\left( -\delta B_{2} |\log h| \right)
\end{eqnarray*}
for \( \delta\geq \delta_{0}>0 \) and some positive constants \( L_{2} \) and \( B_{2}. \)
Furthermore, using the inequality \( |\E_{\P_{t_{r}|t_{0}}} [\Delta_{u,m}^{(2)}]|\leq A_{3}h^{\beta}, \)
we arrive at
\begin{eqnarray*}
\P_{t_{r}|t_{0}}\left( \sup_{x\in \mathcal{A}}|\varepsilon_{M,u}(x)|\geq
 \gamma_{0}\delta\sqrt{|\log h|/(Mh^{d}N_{\beta})} \right)
\leq  L_{3}\exp\left( -\delta B_{3} |\log h| \right)
\end{eqnarray*}
with some positive constants \( L_{3} \) and \( B_{3} \), provided that
\( 6\gamma_{0}^{-1}N^{1/2}_{\beta}
A_{3}h^{\beta}\leq \delta\sqrt{|\log h|/Mh^{d}}.\)
\end{proof}

\section{Appendix}

\subsection{Some results from the theory of empirical processes}
\paragraph{Definition}
A class \( \mathcal{F} \) of functions on a measurable space \( (X,\mathcal{X}) \)
is called a bounded Vapnik-\v{C}ervonenkis class of functions if there exist positive
numbers \( A \) and \( \omega \) such that, for any probability measure \( \P \)
on \( (X,\mathcal{X}) \) and any \( 0<\rho<1 \)
\begin{eqnarray}
\label{VC}
\mathcal{N}(\mathcal{F},L_{2}(\P),\rho\| F \|_{L_{2}(\P)})\leq
\left( \frac{A}{\rho} \right)^{\omega},
\end{eqnarray}
where \( \mathcal{N}(S,d,\varepsilon) \) denotes the \( \varepsilon \)-covering number
of \(S\) in a metric \( d \), and \( F:=\sup_{f\in \mathcal{F}}|f| \) is the envelope
of \( \mathcal{F} \). The following proposition is a key tool
for obtaining convergence
rates for local type estimators.
\begin{prop}[\citet{T}, \citet{GG}]
\label{EIP}
Let \( \mathcal{F} \) be a measurable uniformly bounded VC class of functions,
and let \( \sigma \) and \( U \) be any numbers such that
\( \sup_{f\in \mathcal{F}}\Var(f)\leq \sigma^{2} \), \( \sup_{f\in \mathcal{F}}
\| f \|_{\infty}\leq U \) and \( 0<\sigma<U/2 \).
Then, there exist a universal constant \( B \) and constants \( C \) and \( L \),
depending only on the VC characteristics \( A \) and \( \omega \) of the class \( \mathcal{F}
\), such that
\begin{eqnarray*}
\E\left[ \sup_{f\in \mathcal{F}}
\left| \sum_{m=1}^{M}(f(X_{m})-\E f(X_{1})) \right| \right]
\leq B\left[ \omega U \log \frac{AU}{\sigma}+\sqrt{\omega}\sqrt{M\sigma^{2}\log \frac{AU}{\sigma}}
\right].
\end{eqnarray*}
If moreover \( \sqrt{M}\sigma\geq C_{1}U
\sqrt{\log(U/\sigma)}, \)  there exist constants  \( L \) and \( C \) which
depend only on the VC characteristics of \( \mathcal{F} \), such that, for all \( \lambda\geq C \) and \( t \) satisfying
\begin{eqnarray*}
    C\sqrt{M}\sigma\sqrt{\log \frac{U}{\sigma}}\leq t \leq \lambda \frac{M\sigma^{2}}{U},
\end{eqnarray*}

\begin{eqnarray*}
\label{EI}
\P\left( \sup_{f\in \mathcal{F}}
\left| \sum_{m=1}^{M}(f(X_{m})-\E f(X_{1})) \right|
>t \right)
\leq  L\exp\left( -\frac{\log(1+\lambda/(4L))}{\lambda L}\frac{t^{2}}{ M \sigma^{2}} \right).
\end{eqnarray*}
\end{prop}
\begin{rem}
\label{DepVCD}
It can be deduced from the proof of Proposition~\ref{EIP}  in \citet{GG} that  constant \( L \) can be taken independent of \( \omega \). The constant \( C \) (and hence \( \lambda \)) in the case of large \( \omega \) can be chosen in the form \( C=\omega C_{0} \) for some constant  \( C_{0} \) not depending on \( \omega  \).
\end{rem}


\begin{thebibliography}{99}

\bibitem[Andersen(2000)]{A} L. Andersen (2000). A simple approach to the pricing of Bermudan
swaptions in the multi-factor Libor Market Model. \textit{Journal of
Computational Finance}, \textbf{3}, 5-32.

\bibitem[Audibert and Tsybakov(2007)]{AT} J.-Y. Audibert  and A. Tsybakov  (2007).
Fast learning rates for plug-in classiffiers under the margin condition. \textit{Annals of
Statistics} \textbf{35}, 608 - 633.

\bibitem[Bass(1997)]{B} R. Bass (1997). \textsl{Diffusions and Elliptic Operators}. Springer.

\bibitem[Belomestny et al.(2006)]{BMS} D. Belomestny, G.N. Milstein and V. Spokoiny (2006). Regression
methods in pricing American and Bermudan options using consumption
processes, \textit{Quantitative Finance}, \textbf{9}(3), 315-327.

\bibitem[Belomestny et al.(2007)]{BBS} D. Belomestny, Ch. Bender and J. Schoenmakers (2007). True upper
bounds for Bermudan products via non-nested Monte Carlo,
\textit{Mathematical Finance}, \textbf{19}(1), 53-71.

\bibitem[Broadie and Glasserman(1997)]{BG} M. Broadie and P. Glasserman (1997). Pricing American-style
securities using simulation. \textit{J. of Economic Dynamics and Control}, \textbf{21}%
, 1323-1352.

\bibitem[Carriere(1996)]{Car} J. Carriere (1996). Valuation of early-exercise price of
options using simulations and nonparametric regression. \textit{Insuarance:
Mathematics and Economics}, \textbf{19}, 19-30.

\bibitem[Cl\'{e}ment, Lamberton and Protter(2002)]{CLP} E. Cl\'{e}ment, D. Lamberton and P. Protter (2002). An analysis
of a least squares regression algorithm for American option pricing. \textit{Finance
and Stochastics}, \textbf{6}, 449-471.

\bibitem[Devroye, Gy\"orfi and Lugosi(1996)]{DGL} L. Devroye, L. Gy\"orfi and G. Lugosi (1996).
A probabilistic theory of pattern recognition. Application of Mathematics (New York), \textbf{31}, Springer.

\bibitem[Dudley(1999)]{Du} R. M. Dudley (1999). Uniform Central Limit Theorems,
Cambridge University Press, Cambridge, UK.


\bibitem[Egloff(2005)]{E} D. Egloff (2005). Monte Carlo algorithms for optimal stopping
and statistical learning. \textit{Ann. Appl. Probab.}, \textbf{15},  1396-1432.

\bibitem[Egloff, Kohler and Todorovic(2007)]{EKT} D. Egloff, M. Kohler and N. Todorovic (2007).
A dynamic look-ahead Monte Carlo algorithm for pricing Bermudan options,
\textit{Ann. Appl. Probab.}, \textbf{17}, 1138-1171.


\bibitem[Friedman(1964)]{F}A. Friedman(1964). \textsl{Partial Differential Equations of Parabolic Type}. Prentice-Hall, Englewood
Cliffts, NJ.

\bibitem[Gin\'e and Guillou(2001)]{GG} E. Gin\'e and A. Guillou (2001).
A law of the iterated logarithm for kernel density estimators in the presence of censoring. \textit{Ann. I. H. Poincar\'e}, \textbf{37}, 503-522.


\bibitem[Glasserman(2004)]{Gl} P. Glasserman (2004). Monte Carlo Methods in Financial
Engineering. Springer.

\bibitem[Glasserman and Yu(2005)]{GY} P. Glasserman and B. Yu (2005).  Pricing American Options by Simulation:
Regression Now or Regression Later?, Monte Carlo and Quasi-Monte Carlo Methods,
(H. Niederreiter, ed.), Springer, Berlin.

\bibitem[Kim and Song(2007)]{KS} P. Kim and R. Song (2007). Estimates
on Green functions and Schr\"edinger-type equations for non-symmetric diffusions
with measure-valued drifts. \textit{J. Math. Anal. Appl.}, \textbf{332}, 57-80.


\bibitem[Kohler, Krzyzak and Todorovic(2009)]{KKT} M. Kohler, A. Krzyzak and N. Todorovic (2009).
Pricing of high-dimensional American options by neural networks. To appear in
Mathematical Finance.

\bibitem[Lamberton and Lapeyre(1996)]{LL} D. Lamberton and B. Lapeyre (1996). Introduction to Stochastic
Calculus Applied to Finance. Chapman \& Hall.

\bibitem[Longstaff and Schwartz(2001)]{LS} F. Longstaff and E. Schwartz (2001). Valuing American options
by simulation: a simple least-squares approach. \textit{Review of Financial Studies},
\textbf{14}, 113-147.


\bibitem[Mammen and Tsybakov(1999)]{MT} E. Mammen and A. Tsybakov (1999). Smooth discrimination analysis. \textit{Ann. Statist.,} \textbf{27}, 1808-1829.




\bibitem[Sato and Ueto(1965)]{SU} K. Sato and T. Ueto (1965). Multi-dimensional diffusion and the Markov process on the boundary.\textit{ J. Math. Kyoto Univ.}, 4-3, 529-605.

\bibitem[Talagrand(1994)]{T} M. Talagrand (1994). Sharper bounds for Gaussian
and empirical processes. \textit{Ann. Probab.}, \textbf{22}, 28-76.


\bibitem[Tsitsiklis and Van Roy(1999)]{TV} J. Tsitsiklis and B. Van Roy (1999). Regression methods for
pricing complex American style options. \textit{IEEE Trans. Neural. Net.}, \textbf{12}%
, 694-703.
\bibitem[Van Roy(2009)]{VR} B. Van Roy (2009). On regression-based stopping times,
forthcoming in \textit{Discrete Event Dynamic Systems}.

\end{thebibliography}
\end{document}